\documentclass[12pt]{amsart}

\usepackage{microtype}%
\usepackage{url}
\usepackage{graphicx}
\usepackage[margin=2.5cm]{geometry}
\usepackage{amsmath}
\usepackage{amsfonts}
\usepackage{latexsym}
\usepackage{amssymb}
\usepackage{enumerate}
\usepackage{tikz}
\usetikzlibrary{shapes,snakes}
\usetikzlibrary{automata,positioning}

\usepackage{hyperref}

\usepackage{todonotes}
\newtheorem{definition}{Definition}[section]
\newtheorem{lemma}[definition]{Lemma}
\newtheorem{proposition}[definition]{Proposition}

\newtheorem{theorem}[definition]{Theorem}

\newtheorem{maintheorem}{Theorem}

\theoremstyle{definition}
\newtheorem{example}[definition]{Example}

\newcommand{\N}{\mathbb{N}}
\newcommand{\Z}{\mathbb{Z}}

\newcommand{\R}{\mathbb{R}}

\newcommand{\Bcal}{\mathcal{B}}

\newcommand{\Fcal}{\mathcal{F}}

\newcommand{\Scal}{\mathcal{S}}
\newcommand{\Tcal}{\mathcal{T}}

\newcommand{\Zcal}{\mathcal{Z}}
\newcommand{\0}{\mathtt{0}}
\newcommand{\1}{\mathtt{1}}
\newcommand{\2}{\mathtt{2}}
\newcommand{\Alsig}{\{\0,\1\}}
\newcommand{\Aldel}{\{\0,\1,\2\}}

\newcommand{\bn}{{\boldsymbol{n}}}

\DeclareMathOperator{\rep}{rep}
\DeclareMathOperator{\pad}{pad}

\DeclareMathOperator{\val}{val}

\DeclareMathOperator{\sumF}{sum}
\DeclareMathOperator{\last}{\rm last}
\DeclareMathOperator{\rad}{\rm rad}
\DeclareMathOperator{\rev}{\rm rev}

\newcommand{\valTwoC}{\val_{2c}}
\newcommand{\valF}{\val_\Fcal}
\newcommand{\valFc}{\val_{\Fcal c}}
\newcommand{\repF}{\rep_\Fcal}
\newcommand{\repFc}{\rep_{\Fcal c}}

\allowdisplaybreaks

\begin{document}

\title{A Fibonacci analogue of the two's complement numeration system}

\author[S.~Labb\'e]{S\'ebastien Labb\'e}
\address[S.~Labb\'e]{Univ. Bordeaux, CNRS, Bordeaux INP, LaBRI, UMR 5800, F-33400 Talence, France}
\email{sebastien.labbe@labri.fr}
\urladdr{http://www.slabbe.org/}

\author[J.~Lep\v{s}ov\'a]{Jana Lep\v{s}ov\'a}
\address[J.~Lep\v{s}ov\'a]{FNSPE, CTU in Prague, Trojanova 13, 120 00 Praha, Czech Republic}
\email{jana.lepsova@labri.fr}


\maketitle

\begin{abstract}
Using the classic two's complement notation of signed integers, the fundamental arithmetic operations of addition, subtraction, and multiplication are identical to those for unsigned binary numbers. We introduce a Fibonacci-equivalent of the two's complement notation and we show that addition in this numeration system can be performed by a deterministic finite-state transducer. The result is based on the Berstel adder, which performs addition of the usual Fibonacci representations of nonnegative integers and for which we provide a new constructive proof. Moreover, we characterize the Fibonacci-equivalent of the two's complement notation as an increasing bijection between $\mathbb{Z}$ and a~particular language.
\end{abstract}

\section{Introduction}

A nonnegative integer can be written as a sum of powers of 2, which gives rise to
its binary representation over the alphabet $\Sigma=\{\0,\1\}$.
Binary representations can be added with a standard algorithm, starting from 
the least significant digit and transferring a carry at each step. 
Provided that the representations differ in length,
the shorter one is padded with a prefix of leading zeroes, as in the following example:
\begin{equation*}
    \begin{aligned}
        &11 \quad&\texttt{01011} \\
        +&17 \quad&\texttt{10001} \\[-.4cm]
        \cline{1-3}
        &28 \quad&\texttt{11100}
    \end{aligned}
\end{equation*}

\subsection*{Two's complement numeration system}
Among all the ways to generalize this approach to~$\Z$
is the two's complement notation; see \cite[\S 4.1]{MR3077153}.
In the two's complement representation of integers,
the value of a binary word $w=w_{k-1} w_{k-2}\cdots w_1w_0\in\Sigma^{k}$
is
\begin{equation}\label{eq:twos-complement-value}
    \valTwoC(w) 
    = - w_{k-1}2^{k} + \sum_{i=0}^{k-1}w_i 2^i.
\end{equation}
The value $\valTwoC(w)$ is congruent modulo $2^k$ to the value 
of $w$ as an unsigned integer.
The right-hand side of Equation~\eqref{eq:twos-complement-value}
can be simplified to $- w_{k-1}2^{k-1}+\sum_{i=0}^{k-2}w_i 2^i $,
which is commonly used.
However, 
as we will see shortly,
Equation~\eqref{eq:twos-complement-value} in the unsimplified form lends itself to
a straightforward generalization using
the Fibonacci numbers.
It can be seen that for every $w\in\Sigma^*$, 
$\valTwoC(\0\0w) = \valTwoC(\0w)$
and 
$\valTwoC(\1\1w) = \valTwoC(\1w)$.
Thus $\0$ and $\1$ are neutral prefixes which
can be used to pad the representations of nonnegative and negative numbers, respectively,
keeping the two's complement value invariant.
For every $n\in\Z$ there exists a unique word 
$w\in\Sigma^+\setminus \left( \0\0\Sigma^* \cup \1\1\Sigma^* \right)$
such that $n=\valTwoC(w)$. 
The word $w$ is called the \emph{two's complement representation} of the integer $n$
and we denote it by $\rep_{2c}(n)$.

In the two's complement notation the fundamental
arithmetic operations of addition, subtraction, and multiplication are
identical to those for unsigned binary representations.
To illustrate this, we perform the addition of the representations shown
previously, this time interpreting them in the two's complement notation.
The first word has the same value $\valTwoC(\0\1\0\1\1)=2^3+2^1+2^0=11$,
whereas the second word evaluates as $\valTwoC(\1\0\0\0\1)=-2^4+2^0=-15$.
We compute:
\begin{equation*}
    \begin{aligned}
        &11  \quad&\texttt{01011} \\
        -&15 \quad&\texttt{10001} \\[-.4cm] 
        \cline{1-3}
        -&4 \quad&\texttt{11100}
    \end{aligned}
\end{equation*}
The value of the resulting word is $\valTwoC(\1\1\1\0\0)=-2^4+2^3+2^2=-4$,
which confirms that the computation is correct. Notice that
the negative integer $-4$ has a shorter two's complement representation
$\rep_{2c}(-4)=\1\0\0$, which no longer contains the neutral prefix.

\subsection*{Fibonacci numeration system}
Integers can also be expressed in other numeration systems
\cite[\S 7]{MR1905123}, see also \cite{MR777556,MR1411227,MR2766740}.
A typical example uses the Fibonacci numbers instead of the powers of 2.
Let $(F_n)_{n\geq 0}$ be the Fibonacci sequence
defined by the recurrence relation
$F_{n} = F_{n-1} + F_{n-2}$, for all $n \geq 2$,
and the initial conditions
$ F_0 = 1$, $F_1 = 2$, following a convention for the Fibonacci numeration
system $\Fcal$ \cite{MR942576}.
In this numeration system,
the value of a binary word $w=w_{k-1} w_{k-2} \cdots w_1  w_0 \in\Sigma^k$ is
\[
\valF(w)=\sum_{i=0}^{k-1}w_iF_i.
\]
A result attributed to Zeckendorf \cite{MR58626,MR112863,MR236094,MR308032}
states that for every nonnegative integer $n$ there exists a unique binary word
$w\in \Sigma^*\setminus(\Sigma^*\1\1\Sigma^*\cup\0\Sigma^*)$ such that $n=\valF(w)$.
Note that this result appeared earlier in a more general form in the
Ostrowski's work \cite{MR3069389}.
In other words, every nonnegative integer
can be uniquely represented
as a sum of nonconsecutive distinct Fibonacci numbers.
For example, we have $20=13+5+2=F_5+F_3+F_1=\valF(\texttt{101010})$.
The unique word
$w\in \Sigma^*\setminus(\Sigma^*\1\1\Sigma^*\cup\0\Sigma^*)$ such that $n=\valF(w)$
is denoted by $\repF(n)$; see Table~\ref{table:repF}. 
We refer to this unique word $\repF(n)$ as to the Fibonacci representation of $n$.
Note that the empty word is the Fibonacci representation of the integer 0 and it is denoted
by $\varepsilon$.

\begin{table}[h]
\begin{center}
    {\scriptsize
\begin{tabular}{|c|r|}
$n$ & $\repF(n)$ \\ \hline
&\\[-3mm]
0  & $\varepsilon$   \\
1  & \texttt{1}      \\
2  & \texttt{10}     \\
3  & \texttt{100}    \\
4  & \texttt{101}    \\
5  & \texttt{1000}   \\
6  & \texttt{1001}   \\
7  & \texttt{1010}   \\
8  & \texttt{10000}  \\
9  & \texttt{10001}  \\
\end{tabular}
\quad
\begin{tabular}{|c|r|}
	$n$ & $\repF(n)$ \\ \hline
	&\\[-3mm]
	10 & \texttt{10010} \\
	11 & \texttt{10100} \\
	12 & \texttt{10101} \\
	13 & \texttt{100000} \\
	14 & \texttt{100001} \\
	15 & \texttt{100010} \\
	16 & \texttt{100100} \\
	17 & \texttt{100101} \\
	18 & \texttt{101000} \\
	19 & \texttt{101001} \\
\end{tabular}
\quad
\begin{tabular}{|c|r|}
	$n$ & $\repF(n)$ \\ \hline
	&\\[-3mm]
	20 & \texttt{101010} \\
	21 & \texttt{1000000} \\
	22 & \texttt{1000001} \\
	23 & \texttt{1000010} \\
	24 & \texttt{1000100} \\
	25 & \texttt{1000101} \\
	26 & \texttt{1001000} \\
	27 & \texttt{1001001} \\
	28 & \texttt{1001010} \\
	29 & \texttt{1010000} \\
\end{tabular}}
\end{center}
    \caption{The Fibonacci numeration system $\Fcal$.}
    \label{table:repF}
\end{table}

One way to extend the Fibonacci numeration system 
to all integers is to use Fibonacci numbers with negative indices.
In \cite{MR1162409}, it was proved that
every integer~$n$ whether positive, negative, or zero,
can be uniquely written 
as a sum of nonconsecutive
distinct Fibonacci numbers $F_i$ where $i< -2$.
Let us stress that in this work, we use $F_0=1$ and $F_1=2$, so that
$F_{-1}=1$,
$F_{-2}=0$,
$F_{-3}=1$,
$F_{-4}=-1$,
$F_{-5}=2$,
$F_{-6}=-3$ and $F_{-i}=(-1)^{i+1}F_{i-4}$.
This system is called the negaFibonacci number system in 
\cite[\S 7.1.3]{MR3444818}, where it is used to navigate efficiently
in a pentagrid within the hyperbolic plane.

The addition of integers (including negative integers) based on Fibonacci
numbers was considered in \cite{MR3093678},
where a bit sign is appended to the 
Fibonacci representations. 
The first bit indicates the sign; more precisely, $\0$ means nonnegative, 
$\1$ indicates nonpositive
and the integer zero has two representations.
Problems arise when adding integers with opposite signs because
the sign of the result depends on which number has the greater magnitude.

\subsection*{A Fibonacci analogue of the two's complement}
An analog of the two's complement numeration system
can be obtained by using Fibonacci numbers instead of powers of 2
in Equation~\eqref{eq:twos-complement-value}.
This numeration system can be defined by
the map $\valFc:\{\0,\1\}^*\to\Z$,
\begin{equation}\label{eq:val-Fcal}
    \valFc(w) 
    =  - w_{k-1}F_{k} + \sum_{i=0}^{k-1}w_i F_{i},
\end{equation}
where $w=w_{k-1}\cdots w_0\in\{\0,\1\}^{k}$.
Note that it simplifies to
$\valFc(w)=-w_{k-1}F_{k-2}+\sum_{i=0}^{k-2}w_i F_{i}$.
Denoting $\Sigma = \{\0,\1\}$,
we show that for every $n\in\Z$ there exists a unique odd-length word 
$w\in\Sigma(\Sigma\Sigma)^*\setminus
\left( \Sigma^*\1\1\Sigma^* \cup \0\0\0\Sigma^* \cup \1\0\1\Sigma^*\right)$
such that $n=\valFc(w)$.
This unique word is denoted by
$\repFc(n)$; see Proposition~\ref{prop:Frep}.
The Fibonacci complement numeration system $\Fcal c$
is defined by the map $\repFc:\Z\to\{\0,\1\}^*$ whose values for small integers
are shown in Table~\ref{table:repFc}.
Note that the numeration system $\Fcal c$ extends naturally to $\Z^d$ with an appropriate
padding of shorter words; see Definition~\ref{def:num-sys-Z2}.

\begin{table}[h]
\begin{center}
    \scriptsize
    \begin{tabular}{|c|r|}
        $n$ & $\repFc(n)$ \\ \hline
        &\\[-3mm]
        $-10$ & \texttt{1000100}   \\
        $-9 $ & \texttt{1000101}    \\
        $-8 $ & \texttt{1001000}    \\
        $-7 $ & \texttt{1001001}    \\
        $-6 $ & \texttt{1001010}    \\
        $-5 $ & \texttt{10000}      \\
        $-4 $ & \texttt{10001}      \\
        $-3 $ & \texttt{10010}      \\
        $-2 $ & \texttt{100}        \\
        $-1 $ & \texttt{1}        \\
    \end{tabular}
\quad
    \begin{tabular}{|c|r|}
        $n$ & $\repFc(n)$ \\ \hline
        &\\[-3mm]
        $0$  & \texttt{0}\\
        $1$  & \texttt{001}     \\
        $2$  & \texttt{010}     \\
        $3$  & \texttt{00100}   \\
        $4$  & \texttt{00101}   \\
        $5$  & \texttt{01000}   \\
        $6$  & \texttt{01001}   \\
        $7$  & \texttt{01010}   \\
        $8$  & \texttt{0010000} \\
        $9$  & \texttt{0010001} \\
     \end{tabular}
     \quad
    \begin{tabular}{|c|r|}
        $n$ & $\repFc(n)$ \\ \hline
        &\\[-3mm]
        $10$ & \texttt{0010010} \\
        $11$ & \texttt{0010100} \\
        $12$ & \texttt{0010101} \\
        $13$ & \texttt{0100000} \\
        $14$ & \texttt{0100001} \\
        $15$ & \texttt{0100010} \\
        $16$ & \texttt{0100100} \\
        $17$ & \texttt{0100101} \\
        $18$ & \texttt{0101000} \\
        $19$ & \texttt{0101001} \\
    \end{tabular}
\end{center}
    \caption{The Fibonacci complement numeration system $\Fcal c$.}
    \label{table:repFc}
\end{table}

The numeration system $\Fcal c$ was discovered while studying
a particular aperiodic set of 16 Wang tiles \cite{MR4364231}.
A valid tiling of the plane by the Wang tiles is produced by a specific automaton 
related to the numeration system $\Fcal c$.
More precisely, this automaton takes the
Fibonacci complement representation of a position $(n,m)\in\Z^2$ as an input and
it returns the tile to place at position $(n,m)$ as an output.
The numeration system $\Fcal c$ appears naturally in this context because
it is related to the Fibonacci substitution \cite{2302.14481}, 
which follows from Theorem~\ref{thm:repF_is_increasing} proved in
Section~\ref{sec:increasing-bijection}. 

\subsection*{Addition in the Fibonacci numeration system $\Fcal$}
Let us recall that, following \cite{zbMATH03947643}, the addition of Fibonacci
representations can be performed in terms of a Mealy machine.
A~\emph{Mealy machine} $M$ is a labeled directed graph whose vertices are called
\emph{states} and edges are called \emph{transitions} \cite[Appendix A]{zbMATH05970650}.
The transitions are labeled by pairs $a/b$ of letters.
The first letter $a\in A$ is the input symbol and the second letter $b\in B$ is the output symbol.
The empty word $\varepsilon$ is sometimes included in $B$.
For every state $s$ and every letter $a$, there is at most one transition starting
from the state $s$ with the input symbol $a$.
One distinguished state is called the \emph{initial state}.

A machine $M$ computes a function $M:A^*\to B^*$.
Let $x=x_0x_1\cdots x_{k-1}\in A^*$ and $y=y_0y_1\cdots y_{k-1}\in B^*$ be two words of length $k\in\N$ over the input and output alphabets.
The word $y$ is the output of $x$ under the machine $M$
if and only if there is a sequence $\{s_i\}_{0\leq i\leq k}$
of $k+1$ states of $M$ such that
$s_0$ is the initial state and
for every $i$ with $0\leq i< k$,
there is a transition from $s_{i}$ to $s_{i+1}$ labeled by $x_i/y_i$.
The output word $y$ is denoted by $M(x)$.
The last state $s_k$ is denoted by $M_{\last}(x)$
and an extra output word depending on the last state is denoted by $M_{\downarrow}(x)$.

A 10-state Mealy machine $\Bcal$,
called \emph{the adder} by Berstel
\cite[p.~22]{zbMATH03947643},
reading from left to right is shown in
Figure~\ref{fig:berstel_transducer}.
The states are pairs made of a binary word of length 3
and an additional integer whose meaning is explained in Section~\ref{sec:proof-Berstel-adder}.
The Berstel adder $\Bcal$ has the property that
for every input $u \in \Aldel^*$, $u$ has the same value in the Fibonacci
numeration system as the concatenation of the output word $\Bcal(u) \in
\Alsig^*$ with the three-letter binary word $\Bcal_{\downarrow}(u)\in\{\0\0\0,
\0\0\1, \0\1\0, \1\0\0, \1\0\1\}$, which depends only on the final state reached when
reading~$u$:
\[
\valF(u) = \valF(\Bcal(u)\cdot\Bcal_{\downarrow}(u)).
\] 
This result is stated as our Theorem~\ref{thm:Berstel} below.
Computations performed with the Berstel adder are shown in
Example~\ref{example:berstel_addition}
and Example~\ref{example:berstel_addition-2}.

\begin{figure}[h!]
\begin{center}
    \includegraphics[width=\linewidth]{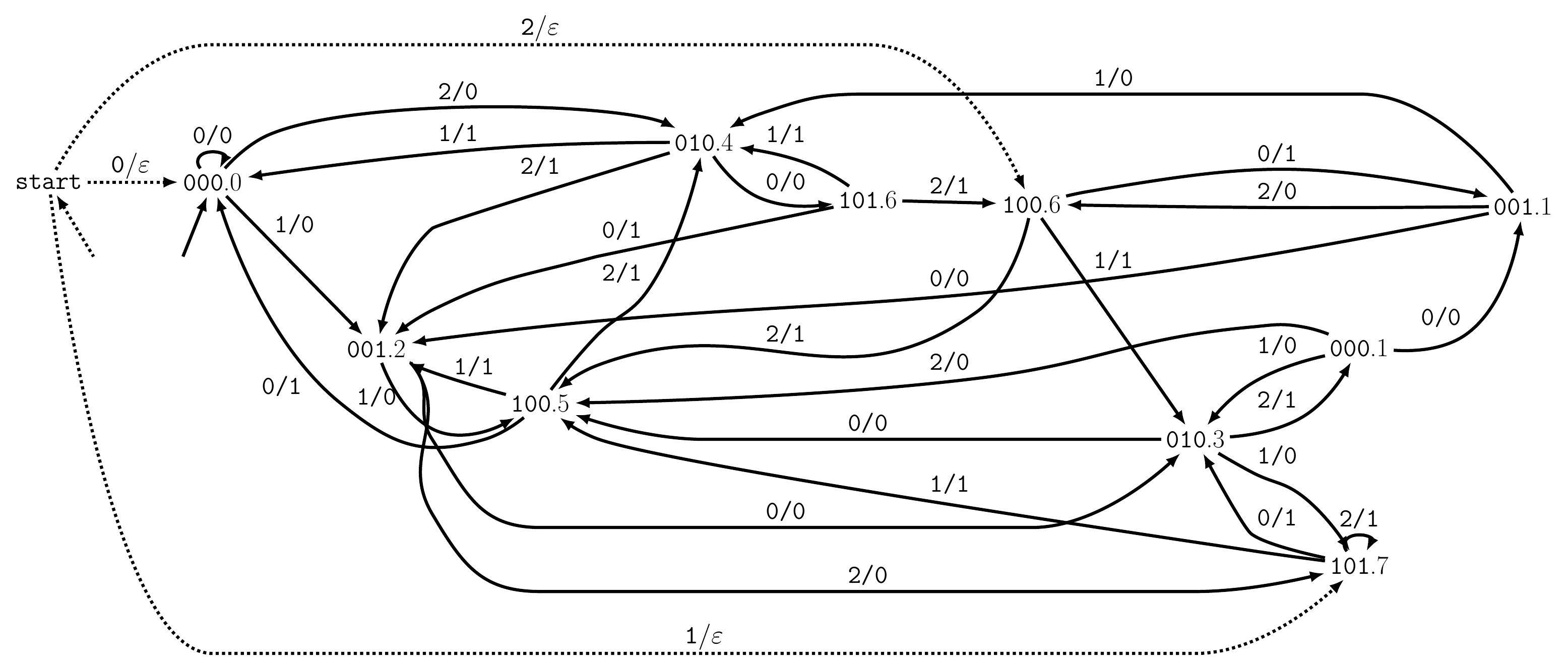}
\end{center}
    \caption{The Berstel adder $\Bcal$ is a 10-state Mealy machine
    with 30 transitions illustrated as solid edges with initial state $\0\0\0.0$.
    The Mealy machine $\Tcal$ is obtained by adding a new state
    \texttt{start} that replaces $\0\0\0.0$ as initial state and adding
    three additional transitions shown with dashed edges.}
\label{fig:berstel_transducer}
\end{figure}

\subsection*{Main result}
The goal of this contribution is to prove that 
the Fibonacci complement numeration system $\Fcal c$ 
behaves like the two's complement in terms of addition of integers.
More precisely, addition of integers in the numeration system $\Fcal c$ 
can be performed as the usual addition of nonnegative integers in the
Fibonacci numeration system $\Fcal$.

The result is based on the Berstel adder $\Bcal$ to which three transitions 
\[
    \texttt{start}\xrightarrow[]{0/\varepsilon}\texttt{000}.0, \quad
    \texttt{start}\xrightarrow[]{1/\varepsilon}\texttt{101}.7, \quad
    \texttt{start}\xrightarrow[]{2/\varepsilon}\texttt{100}.6
\]
are added from a new initial state $\texttt{start}$ that replaces 
$\texttt{000}.0$
as initial state.
The three additional transitions are shown with dashed edges in
Figure~\ref{fig:berstel_transducer}. We let $\Tcal$ denote this modified Berstel
adder consisting of 11 states and 33 transitions. 
The finite-state transducer $\Tcal$ is a Mealy machine if the empty
word $\varepsilon$ is included in the output alphabet.

\begin{maintheorem}\label{theorem:main}
    The Mealy machine $\Tcal$ has the property that
    for every nonempty input $u \in \Aldel^k$, it outputs a binary word 
    $\Tcal(u)\cdot \Tcal_{\downarrow}(u) \in \Alsig^+$ 
    of length $k+2$
    with the same value for the Fibonacci complement numeration system, i.e.,
    \[
    \valFc(u) = \valFc(\Tcal(u)\cdot \Tcal_{\downarrow}(u)).
    \]
\end{maintheorem}

Computations illustrating Theorem~\ref{theorem:main}
are shown in Example~\ref{example:Tcal_addition}
and Example~\ref{example:Tcal_addition-2}.

This result suggests a possible generalization to other numeration systems,
starting with
simple Parry numeration systems which include the Fibonacci numeration system; see 
\cite{MR97374,
MR166332,
MR2766740,
MR3933318}.
This remains to be explored.

\textbf{Structure of the article.}
In Section~\ref{sec:preliminaries}, we recall a few preliminary results
on the numeration system $\Fcal c$ for $\Z$.
In Section~\ref{sec:addtion-Zeck-N}, we recall the addition of Fibonacci
representations on $\N$. 
In Section~\ref{sec:addtion-Fcal-Z},
we prove Theorem~\ref{theorem:main}.
In Section~\ref{sec:increasing-bijection},
we show that the map $\repFc$ defining the Fibonacci complement numeration
system is characterized by the fact of being an increasing bijection
for some total order defined on a particular language;
see Theorem~\ref{thm:repF_is_increasing}.
In Section~\ref{sec:proof-Berstel-adder},
we provide a constructive proof that
the Berstel adder $\Bcal$ performs addition of nonnegative integers.

\subsection*{Acknowledgments}
The first author is partially funded by France's
ANR CODYS (ANR-18-CE40-0007)
and ANR IZES (ANR-22-CE40-0011).
The second author acknowledges support from Barrande Fellowship Programme 
and from Grant Agency of Czech Technical University
in Prague, through the project SGS20/183/OHK4/3T/14.
We are thankful to Christiane Frougny and Jacques Sakarovitch
for helpful discussions on numeration systems.
We thank Julien Cassaigne to make us realize that Equations
\eqref{eq:twos-complement-value} and \eqref{eq:val-Fcal} are really analogous
before simplification (this is not the case after simplification).
We thank Jeffrey Shallit for making us aware of \cite{MR3518158}, as well as
earlier references \cite{MR3069389,MR58626} about Zeckendorf's theorem.
We would like to thank the referees for their in-depth reading of the article
and their many valuable comments.

\section{A Fibonacci Complement Numeration System for $\Z$}\label{sec:preliminaries}

In this section, we present the Fibonacci complement numeration system,
which is
defined by the value map $\valFc:\Sigma^*\to\Z$ 
given in Equation~\eqref{eq:val-Fcal},
where $\Sigma=\Alsig$.

The first observation to make on this value map is given in the next lemma.

\begin{lemma}\label{lem:neutral_prefix}
    For every word $w\in\Sigma^{*}$,
    we have
    \[
        \valFc(\0\0\0w)=\valFc(\0w)
        \quad
        \text{ and }
        \quad
        \valFc(\1\0\1w)=\valFc(\1w).
    \]
\end{lemma}

\begin{proof}
    Let $w=w_{k-1} \cdots w_0\in\Sigma^{*}$.
    We have
    $\valFc(\0\0\0w) 
        = \sum_{i=0}^{k-1}w_i F_i
        = \valFc(\0w)$.
    Also,
    \[
       \valFc(\1\0\1w) 
        = - F_{k+1}  + F_{k} + \sum_{i=0}^{k-1}w_i F_i
        = - F_{k-1} + \sum_{i=0}^{k-1}w_i F_i
        = \valFc(\1w).\qedhere
    \]
\end{proof}

Thus $\0\0$ or $\1\0$ can be used to pad words without changing their value,
leading to the following definition.

\begin{definition}[Neutral prefix]\label{def:neutral-prefix}
Let $w\in\Sigma^{*}$. 
    We say that $\0\0$ (resp., $\1\0$) 
    is the \emph{neutral prefix} of $w$ if 
    $w\in\0\Sigma^*$ (resp., if $w\in\1\Sigma^*$).
    The neutral prefix is denoted by $p_w$.
\end{definition}

Representations in the numeration system $\Fcal c$ of integers $n\in [-13,21)$
are shown in Figure~\ref{fig:numeration-system-F}. The neutral prefixes
correspond to two loops on the vertices $0$ and $-1$.
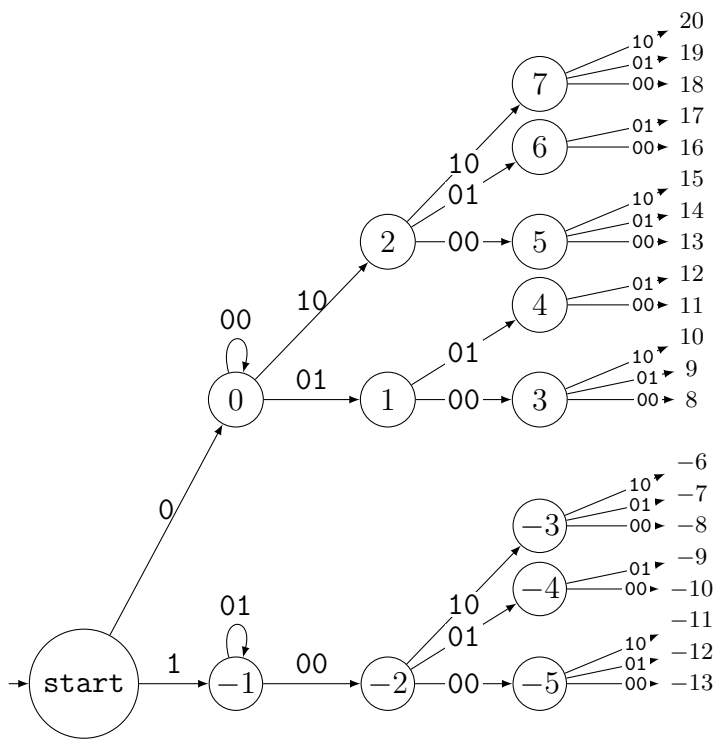
\begin{figure}[h]
\begin{center}
    \begin{tikzpicture}[>=latex]
    \begin{scope}[xscale=1.01,yscale=.42]
        \node (S)  [draw,circle] at (0,-8) {\texttt{start}};
    \node (0)  [draw,circle] at (2,1) {0};
    \node (-1) [draw,circle,inner sep=1pt] at (2,-8) {$-1$};
    \node (1)  [draw,circle] at (4,1) {1};
    \node (2)  [draw,circle] at (4,6) {2};
    \node (-2) [draw,circle,inner sep=1pt] at (4,-8) {$-2$};
    \node (3)  [draw,circle] at (6,1)  {3};
    \node (4)  [draw,circle] at (6,4)  {4};
    \node (5)  [draw,circle] at (6,6)  {5};
    \node (6)  [draw,circle] at (6,9)  {6};
    \node (7)  [draw,circle] at (6,11) {7};
    \node (-3) [draw,circle,inner sep=1pt] at (6,-3) {$-3$};
    \node (-4) [draw,circle,inner sep=1pt] at (6,-5) {$-4$};
    \node (-5) [draw,circle,inner sep=1pt] at (6,-8) {$-5$};
    \node (8)    at (8,1)  {\scriptsize 8};
    \node (9)    at (8,2)  {\scriptsize 9};
    \node (10)   at (8,3)  {\scriptsize 10};
    \node (11)   at (8,4)  {\scriptsize 11};
    \node (12)   at (8,5)  {\scriptsize 12};
    \node (13)   at (8,6)  {\scriptsize 13};
    \node (14)   at (8,7)  {\scriptsize 14};
    \node (15)   at (8,8)  {\scriptsize 15};
    \node (16)   at (8,9)  {\scriptsize 16};
    \node (17)   at (8,10) {\scriptsize 17};
    \node (18)   at (8,11) {\scriptsize 18};
    \node (19)   at (8,12) {\scriptsize 19};
    \node (20)   at (8,13) {\scriptsize 20};
    \node (-6)   at (8,-1) {\scriptsize $-6$};
    \node (-7)   at (8,-2) {\scriptsize $-7$};
    \node (-8)   at (8,-3) {\scriptsize $-8$};
    \node (-9)   at (8,-4) {\scriptsize $-9$};
    \node (-10)  at (8,-5) {\scriptsize $-10$};
    \node (-11)  at (8,-6) {\scriptsize $-11$};
    \node (-12)  at (8,-7) {\scriptsize $-12$};
    \node (-13)  at (8,-8) {\scriptsize $-13$};
    \end{scope}
    \draw[loop above,->] (0) to node{\tt 00} (0);
    \draw[loop above,->] (-1) to node{\tt 01} (-1);
    \draw[<-] (S)  -- ++ (-1,0);
    \draw[->] (S)  -- node[above]{\tt 0} (0);
    \draw[->] (S)  -- node[above]{\tt 1} (-1);
    \draw[->] (0)  -- node[above]{\tt 10} (2);
    \draw[->] (0)  -- node[above]{\tt 01} (1);
    \draw[->] (-1) -- node[above]{\tt 00} (-2);
    \draw[->] (2)  -- node[fill=white,inner sep=1pt]{\tt 10} (7);
    \draw[->] (2)  -- node[fill=white,inner sep=1pt]{\tt 01} (6);
    \draw[->] (2)  -- node[fill=white,inner sep=1pt]{\tt 00} (5);
    \draw[->] (1)  -- node[fill=white,inner sep=1pt]{\tt 01} (4);
    \draw[->] (1)  -- node[fill=white,inner sep=1pt]{\tt 00} (3);
    \draw[->] (-2) -- node[fill=white,inner sep=1pt]{\tt 10} (-3);
    \draw[->] (-2) -- node[fill=white,inner sep=1pt]{\tt 01} (-4);
    \draw[->] (-2) -- node[fill=white,inner sep=1pt]{\tt 00} (-5);
    \draw[->] (7)  -- node[fill=white,near end,inner sep=1pt]{\tiny\tt 10} (20);
    \draw[->] (7)  -- node[fill=white,near end,inner sep=1pt]{\tiny\tt 01} (19);
    \draw[->] (7)  -- node[fill=white,near end,inner sep=1pt]{\tiny\tt 00} (18);
    \draw[->] (6)  -- node[fill=white,near end,inner sep=1pt]{\tiny\tt 01} (17);
    \draw[->] (6)  -- node[fill=white,near end,inner sep=1pt]{\tiny\tt 00} (16);
    \draw[->] (5)  -- node[fill=white,near end,inner sep=1pt]{\tiny\tt 10} (15);
    \draw[->] (5)  -- node[fill=white,near end,inner sep=1pt]{\tiny\tt 01} (14);
    \draw[->] (5)  -- node[fill=white,near end,inner sep=1pt]{\tiny\tt 00} (13);
    \draw[->] (4)  -- node[fill=white,near end,inner sep=1pt]{\tiny\tt 01} (12);
    \draw[->] (4)  -- node[fill=white,near end,inner sep=1pt]{\tiny\tt 00} (11);
    \draw[->] (3)  -- node[fill=white,near end,inner sep=1pt]{\tiny\tt 10} (10);
    \draw[->] (3)  -- node[fill=white,near end,inner sep=1pt]{\tiny\tt 01}  (9);
    \draw[->] (3)  -- node[fill=white,near end,inner sep=1pt]{\tiny\tt 00}  (8);
    \draw[->] (-3) -- node[fill=white,near end,inner sep=1pt]{\tiny\tt 10} (-6);
    \draw[->] (-3) -- node[fill=white,near end,inner sep=1pt]{\tiny\tt 01} (-7);
    \draw[->] (-3) -- node[fill=white,near end,inner sep=1pt]{\tiny\tt 00} (-8);
    \draw[->] (-4) -- node[fill=white,near end,inner sep=1pt]{\tiny\tt 01} (-9);
    \draw[->] (-4) -- node[fill=white,near end,inner sep=1pt]{\tiny\tt 00} (-10);
    \draw[->] (-5) -- node[fill=white,near end,inner sep=1pt]{\tiny\tt 10} (-11);
    \draw[->] (-5) -- node[fill=white,near end,inner sep=1pt]{\tiny\tt 01} (-12);
    \draw[->] (-5) -- node[fill=white,near end,inner sep=1pt]{\tiny\tt 00} (-13);
    \end{tikzpicture}
\end{center}
\caption{
The language $\repFc(\Z)$ can be represented as a tree of integers where each 
word $\repFc(n)$ labels the path from the node \texttt{start} to the node $n\in\Z$.
The labeling of the nodes by integers illustrates a breadth-first traversal of the
upper and lower half of the tree.}
\label{fig:numeration-system-F}
\end{figure}

Note the following relation between
the Fibonacci numeration system
and the Fibonacci complement numeration system:
    \begin{equation}\label{eq:valFc-vs-valF}
    \valFc(w) 
    =  - w_{k-1}F_{k} + \sum_{i=0}^{k-1}w_i F_{i}
    =  - w_{k-1}F_{k} + \valF(w).
    \end{equation}
    for every nonempty word $w=w_{k-1}\cdots w_0 \in \Sigma^+$ of length $k$.

The following proposition enables us to define representations of integers in the
Fibonacci complement numeration system.
Its proof is in Section~\ref{sec:increasing-bijection}.

\begin{proposition}\label{prop:Frep}
    The map  
    $\valFc\colon \Sigma(\Sigma\Sigma)^*\setminus
    \left( \Sigma^*\1\1\Sigma^* \cup \0\0\0\Sigma^* \cup \1\0\1\Sigma^*\right)
    \to \Z$ is a bijection.
\end{proposition}

The inverse of the map $\valFc$ defines the Fibonacci complement numeration system.

\begin{definition}[Fibonacci complement numeration system]\label{def:num-sys-Z}
    For every $n\in\Z$, we let $\repFc(n)$ denote the unique word
    $w\in \Sigma(\Sigma\Sigma)^*\setminus
    \left( \Sigma^*\1\1\Sigma^* \cup \0\0\0\Sigma^* \cup \1\0\1\Sigma^*\right)$
    such that $\valFc(w)=n$.
\end{definition}

The Fibonacci complement $(\Fcal c)$ numeration system is illustrated in
Figure~\ref{fig:numeration-system-F}.
Note that this numeration system was denoted by the letter $\Fcal$
in \cite{MR4364231}.
Herein, we made the choice of denoting the Fibonacci complement numeration system by
$\Fcal c$ and the Fibonacci numeration system by~$\Fcal$.

The neutral prefix can be used to pad words so that they all have the same length.

\begin{definition}[Pad function]
    Let $u,v\in\Sigma(\Sigma\Sigma)^*$ be two words of odd length.
    We define
    \[
        \pad\begin{pmatrix} u \\ v \end{pmatrix}=\left(
        \begin{array}{c}
            \pad_k(u)\\
            \pad_k(v)
        \end{array}
        \right)
    \]
    where $k=\max\{|u|, |v|\}$
    and 
    $\pad_k(w)= {p_{w}}^{\frac{1}{2}(k-|w|)}w$
    for every $w\in\{u,v\}$
    where $p_w$ is the neutral prefix of the word $w$.
    Note that $k-|w|$ is always even since $u$ and $v$ are of odd length.
\end{definition}

Padding allows us
to represent coordinates
in $\Z^d$ in dimension $d\geq1$. Here we consider the case $d=2$.
Padding also allows us 
to define the sum of words.

\begin{definition}[Numeration system $\Fcal c$ for $\Z^2$]\label{def:num-sys-Z2}
	Let $\bn=(n_1,n_2)\in\Z^2$.
	We define
	\[
	\repFc(\bn)=\pad\left(
	\begin{array}{c}
	\repFc(n_1)\\
	\repFc(n_2)
	\end{array}
	\right).
	\]
\end{definition}

\begin{definition}[Sum of two words]\label{def:addition-Fcal}
    Let $\Sigma = \Alsig$ and $u,v \in \Sigma^*$.
    Then we define $\sumF:\Sigma^*\times \Sigma^*\to \Aldel^*$ as
    \[
        \sumF(u,v)=(u_{k-1}+v_{k-1})\cdots (u_0+v_0)
        \quad
        \text{ where }
        \begin{pmatrix} u_{k-1}\cdots u_0 \\ v_{k-1}\cdots v_0 \end{pmatrix}
            =
        \pad\begin{pmatrix} u \\ v \end{pmatrix}.
    \]
\end{definition}

As the sum of two binary words is over the alphabet $\Aldel$,
we extend the maps $\val_\Fcal$ and $\val_{\Fcal c}$ naturally to the 
alphabet $\Aldel$. The following lemma is analogous to
Lemma~\ref{lem:neutral_prefix}.
\begin{lemma}\label{lem:neutral_prefix_2}
	For every word $v\in\Aldel^{*}$ and every $a\in\Aldel$,
	we have
	\[
	\valFc(a\0 a v)=\valFc(a v).
	\]
\end{lemma}

\begin{proof}
	Let $v=v_{k-1} \cdots v_0\in\Aldel^{*}$ and $a\in\Aldel$.
	We have
	\[
	\valFc(a\0av) 
	=  - aF_{k+1} + aF_{k} + \sum_{i=0}^{k-1}v_i F_i
	=  - aF_{k-1} + \sum_{i=0}^{k-1}v_i F_i
	= \valFc(av).\qedhere
	\]
\end{proof}

\section{Addition of Fibonacci representations on $\N$}\label{sec:addtion-Zeck-N}

In this section, we recall how the addition of Fibonacci representations of
nonnegative integers can be performed with the Berstel adder, 
a finite-state deterministic transducer proposed by Berstel in
\cite[p.~22]{zbMATH03947643}.

We first present the formal definition of Mealy machine
as the Berstel adder is an instance of it.
A detailed introduction to automata theory and finite-state transducers 
can be found for instance in \cite{MR2567276,zbMATH05970650}.
A \emph{Mealy machine with output} is a 7-tuple $M = (S, S_0, A, B, \delta, \eta, \phi)$, where
$S$ is a finite set of states,
$S_0\in S$ is the initial state,
$A$ is the input alphabet,
$B$ is the output alphabet,
$\delta: S\times A \to S$ is
the transition function
mapping pairs of a state and an input symbol to the corresponding next state,
$\eta: S\times A \to B$ is the output function mapping pairs of a state 
and an input symbol to the corresponding output symbol
and
$\phi: S\to B^*$ is the function mapping each state to an extra output word 
over the alphabet $B$.
The empty word $\varepsilon$ is sometimes included in $B$.
Reading a word $u = u_0 \cdots u_{k-1} \in A^*$ for some $k\in\N$, 
the Mealy machine $M$
moves between states $s_i\in S$, with $s_0 = S_0$ and $s_{i+1} = \delta(s_i, u_i)$,
outputting sequentially
one letter $w_i = \eta(s_i, u_i)\in B$ for each input letter $u_i\in A$,
$i\in\{0,1,\dots,k-1\}$.
The output word $w = w_0 \cdots w_{k-1}$ is denoted by $M(u)$.
The last state $s_k$ is denoted by $M_{\last}(u)$ 
and an extra output word depending on the last state is $M_{\downarrow}(u)=\phi(M_{\last}(u))$.
In this article, the extra output word needs to be concatenated after the output word $M(u)$.

The Berstel adder is the Mealy machine $\Bcal= (Q, \0\0\0.0, \Aldel, \Alsig, \delta_\Bcal, \eta_\Bcal, \phi_\Bcal)$ with the set of 10 states 
\[
Q=\left\{
\begin{array}{l}
\mathtt{000}.0,\quad
\mathtt{001}.1,\quad
\mathtt{010}.3,\quad
\mathtt{100}.5,\quad
\mathtt{101}.6,\\
\mathtt{000}.1,\quad
\mathtt{001}.2,\quad
\mathtt{010}.4,\quad
\mathtt{100}.6,\quad
\mathtt{101}.7
\end{array}
\right\}
\]
with initial state $\0\0\0.0$, input alphabet $\Aldel$,
output alphabet $\Alsig$, transition function $\delta_\Bcal$ and output
function $\eta_\Bcal$ as shown in Figure~\ref{fig:berstel_transducer}.
The states of the Berstel adder 
is a subset of $\Scal\times\{0,1,\dots,7\}$
where $\Scal=\{\0\0\0, \0\0\1, \0\1\0, \1\0\0, \1\0\1\}$.
The function $\phi_\Bcal$ is the canonical projection
$\Scal\times\{0,1,\dots,7\}\to \Scal$.
Our representation of the 10 states (in particular the value in $\{0,1,\dots,7\}$) 
is not the same but is equivalent to the 10 states provided by Berstel.
The value in $\{0,1,\dots,7\}$ is used in Section~\ref{sec:proof-Berstel-adder}
to prove that there are finitely many states to consider up to an equivalence relation.
Reading a word $u \in \Aldel^*$,
the Berstel adder outputs $\Bcal(u)\cdot \Bcal_{\downarrow}(u)$.

Berstel introduced the adder $\Bcal$ in \cite{zbMATH03947643}
for the addition of Fibonacci representations of nonnegative integers in $\N$.

\begin{theorem}\label{thm:Berstel}
    {\rm\cite{zbMATH03947643}}
    The Berstel adder $\Bcal$ has the property that
    for every input $u \in \Aldel^*$, it outputs a word $\Bcal(u)\cdot \Bcal_{\downarrow}(u) \in \Alsig^*$ 
    with same value for the Fibonacci numeration system:
    $$
    \valF(u) = \valF(\Bcal(u)\cdot \Bcal_{\downarrow}(u)).
    $$
\end{theorem}

\begin{table}[t]
\[
    \footnotesize
\begin{array}{c||c|c|c||c|c|c}
    u & \valF(u) & \Bcal(u)\Bcal_{\downarrow}(u) & {\valF(\Bcal(u)\Bcal_{\downarrow}(u))} &
    \valFc(u) & \Tcal(u)\Tcal_{\downarrow}(u) & {\valFc(\Tcal(u)\Tcal_{\downarrow}(u))} \\ 
\hline
\hline
\texttt{0 }  & 0  &\texttt{ 0}\cdot\texttt{000}  &   0  &  0  &\varepsilon\cdot\texttt{000} & 0  \\
\texttt{1 }  & 1  &\texttt{ 0}\cdot\texttt{001}  &   1  &  -1 &\varepsilon\cdot\texttt{101} & -1 \\
\texttt{2 }  & 2  &\texttt{ 0}\cdot\texttt{010}  &   2  &  -2 &\varepsilon\cdot\texttt{100} & -2 \\
\hline
\texttt{00}  & 0  &\texttt{00}\cdot\texttt{000}  &   0  &  0  &\texttt{0}\cdot\texttt{000} & 0  \\
\texttt{01}  & 1  &\texttt{00}\cdot\texttt{001}  &   1  &  1  &\texttt{0}\cdot\texttt{001} & 1  \\
\texttt{02}  & 2  &\texttt{00}\cdot\texttt{010}  &   2  &  2  &\texttt{0}\cdot\texttt{010} & 2  \\
\texttt{10}  & 2  &\texttt{00}\cdot\texttt{010}  &   2  &  -1 &\texttt{1}\cdot\texttt{010} & -1 \\
\texttt{11}  & 3  &\texttt{00}\cdot\texttt{100}  &   3  &  0  &\texttt{1}\cdot\texttt{100} & 0  \\
\texttt{12}  & 4  &\texttt{00}\cdot\texttt{101}  &   4  &  1  &\texttt{1}\cdot\texttt{101} & 1  \\
\texttt{20}  & 4  &\texttt{00}\cdot\texttt{101}  &   4  &  -2 &\texttt{1}\cdot\texttt{001} & -2 \\
\texttt{21}  & 5  &\texttt{01}\cdot\texttt{000}  &   5  &  -1 &\texttt{1}\cdot\texttt{010} & -1 \\
\texttt{22}  & 6  &\texttt{01}\cdot\texttt{001}  &   6  &  0  &\texttt{1}\cdot\texttt{100} & 0  \\
\hline
\texttt{000} & 0  &\texttt{000}\cdot\texttt{000} &   0  &  0 & \texttt{00}\cdot \texttt{000} & 0 \\
\texttt{001} & 1  &\texttt{000}\cdot\texttt{001} &   1  &  1 & \texttt{00}\cdot \texttt{001} & 1 \\
\texttt{002} & 2  &\texttt{000}\cdot\texttt{010} &   2  &  2 & \texttt{00}\cdot \texttt{010} & 2 \\
\texttt{010} & 2  &\texttt{000}\cdot\texttt{010} &   2  &  2 & \texttt{00}\cdot \texttt{010} & 2 \\
\texttt{011} & 3  &\texttt{000}\cdot\texttt{100} &   3  &  3 & \texttt{00}\cdot \texttt{100} & 3 \\
\texttt{012} & 4  &\texttt{000}\cdot\texttt{101} &   4  &  4 & \texttt{00}\cdot \texttt{101} & 4 \\
\texttt{020} & 4  &\texttt{000}\cdot\texttt{101} &   4  &  4 & \texttt{00}\cdot \texttt{101} & 4 \\
\texttt{021} & 5  &\texttt{001}\cdot\texttt{000} &   5  &  5 & \texttt{01}\cdot \texttt{000} & 5 \\
\texttt{022} & 6  &\texttt{001}\cdot\texttt{001} &   6  &  6 & \texttt{01}\cdot \texttt{001} & 6 \\
\texttt{100} & 3  &\texttt{000}\cdot\texttt{100} &   3  & -2 & \texttt{10}\cdot \texttt{100} & -2 \\
\texttt{101} & 4  &\texttt{000}\cdot\texttt{101} &   4  & -1 & \texttt{10}\cdot \texttt{101} & -1 \\
\texttt{102} & 5  &\texttt{001}\cdot\texttt{000} &   5  &  0 & \texttt{11}\cdot \texttt{000} & 0 \\
\texttt{110} & 5  &\texttt{001}\cdot\texttt{000} &   5  &  0 & \texttt{11}\cdot \texttt{000} & 0 \\
\texttt{111} & 6  &\texttt{001}\cdot\texttt{001} &   6  &  1 & \texttt{11}\cdot \texttt{001} & 1 \\
\texttt{112} & 7  &\texttt{001}\cdot\texttt{010} &   7  &  2 & \texttt{11}\cdot \texttt{010} & 2 \\
\texttt{120} & 7  &\texttt{001}\cdot\texttt{010} &   7  &  2 & \texttt{11}\cdot \texttt{010} & 2 \\
\texttt{121} & 8  &\texttt{001}\cdot\texttt{100} &   8  &  3 & \texttt{11}\cdot \texttt{100} & 3 \\
\texttt{122} & 9  &\texttt{001}\cdot\texttt{101} &   9  &  4 & \texttt{11}\cdot \texttt{101} & 4 \\
\texttt{200} & 6  &\texttt{001}\cdot\texttt{001} &   6  & -4 & \texttt{10}\cdot \texttt{001} & -4 \\
\texttt{201} & 7  &\texttt{001}\cdot\texttt{010} &   7  & -3 & \texttt{10}\cdot \texttt{010} & -3 \\
\texttt{202} & 8  &\texttt{001}\cdot\texttt{100} &   8  & -2 & \texttt{10}\cdot \texttt{100} & -2 \\
\texttt{210} & 8  &\texttt{010}\cdot\texttt{000} &   8  & -2 & \texttt{10}\cdot \texttt{100} & -2 \\
\texttt{211} & 9  &\texttt{010}\cdot\texttt{001} &   9  & -1 & \texttt{10}\cdot \texttt{101} & -1 \\
\texttt{212} & 10 &\texttt{010}\cdot\texttt{010} &   10 &  0 & \texttt{11}\cdot \texttt{000} & 0 \\
\texttt{220} & 10 &\texttt{010}\cdot\texttt{010} &   10 &  0 & \texttt{11}\cdot \texttt{000} & 0 \\
\texttt{221} & 11 &\texttt{010}\cdot\texttt{100} &   11 &  1 & \texttt{11}\cdot \texttt{001} & 1 \\
\texttt{222} & 12 &\texttt{010}\cdot\texttt{101} &   12 &  2 & \texttt{11}\cdot \texttt{010} & 2 \\
\end{array}
\]
    \caption{The table lists for every ternary word $u\in\{\0,\1,\2\}^+$ of length $\leq3$
    the output $\Bcal(u)\cdot\Bcal_{\downarrow}(u)\in\{\0,\1\}^*$ under the Berstel adder $\Bcal$
    and the output $\Tcal(u)\cdot\Tcal_{\downarrow}(u)\in\{\0,\1\}^*$ under the
    modified Berstel adder $\Tcal$.
    The other columns illustrate that $\Bcal$ preserves the
    value in the Fibonacci numeration system whereas 
    $\Tcal$ preserves the 
    value in the Fibonacci complement numeration system.}
    \label{table:Berstel-adder-length-3-word}
\end{table}

The behavior of the Berstel adder on all ternary words of length $\leq 3$ is
listed in Table~\ref{table:Berstel-adder-length-3-word}.
Below, we illustrate Theorem~\ref{thm:Berstel} on two examples.

\begin{example}\label{example:berstel_addition}
	Feeding the Berstel adder $\Bcal$ in Figure~\ref{fig:berstel_transducer} with the word $u=\2\2\2\0\1\2\1$ gives
	\[
	\0\0\0.0 \xrightarrow{\2/\0} \0\1\0.4\xrightarrow{\2/\1} \0\0\1.2 \xrightarrow{\2/\0}
	\1\0\1.7 \xrightarrow{\0/\1} \0\1\0.3 \xrightarrow{\1/\0} \1\0\1.7
	\xrightarrow{\2/\1} \1\0\1.7 \xrightarrow{\1/\1} \1\0\0.5;
	\]
	therefore the last state is $\Bcal_{\last}(u)=\1\0\0.5$ and we obtain $\Bcal(u)\cdot \Bcal_{\downarrow}(u) = \0\1\0\1\0\1\1\ \cdot \1\0\0$. Thus
	\[
		\valF(\2\2\2\0\1\2\1) = 42+26+16+3+4+1 = 92 = 55+21+8+5+3 = \valF(\0\1\0\1\0\1\1\1\0\0).
	\]
\end{example}

\begin{example}\label{example:berstel_addition-2}
Here is how Berstel adder can be used to compute the sum $33+25$.
First, we express 33 and 25 by their Fibonacci representation
(in general,
if representations do not have the same length, the shorter one is padded
with leading $\0$'s).
Then, we add them digit by digit to obtain a ternary word in $\{\0,\1,\2\}^*$:
\begin{equation*}
\begin{aligned}
&33   \quad&\texttt{1010101} & \\
+&25  \quad&\texttt{1000101} & \\[-.4cm]
\cline{1-3}
&58   \quad&\texttt{2010202} &
\end{aligned}
\end{equation*}
Reading from left to right and giving the word $u=\texttt{2010202}$ as input to
the Berstel adder (see Figure~\ref{fig:berstel_transducer}), we obtain
the following path from the initial state $\texttt{000}.0$:
\[
\texttt{000}.0 \xrightarrow{\2/\0}  \texttt{010}.4
\xrightarrow{\0/\0}  \texttt{101}.6 
\xrightarrow{\1/\1}  \texttt{010}.4
\xrightarrow{\0/\0}  \texttt{101}.6
\xrightarrow{\2/\1}  \texttt{100}.6
\xrightarrow{\0/\1}  \texttt{001}.1
\xrightarrow{\2/\0}  \texttt{100}.6.
\]
Therefore, the output word is $\Bcal(u)=\texttt{0010110}$
and the path ends in state $\Bcal_{\last}(u)=\texttt{100}.6$.
Removing the last digit (6) of the last state, which we can ignore for now
(see Definition~\ref{def:Delta-map} for details),
we obtain the three-letter extra output word $\Bcal_{\downarrow}(u)=\texttt{100}$.
Concatenating the two words gives
$\texttt{0010110}\cdot\texttt{100}$,
which has the correct Fibonacci value 33+25=58
\[
\valF\left(\Bcal(u)\cdot\Bcal_{\downarrow}(u)\right)
= \valF\left(\texttt{0010110}\cdot\texttt{100}\right)
= 3+8+13+34=58;
\]
however, it is not the normal representation as it contains leading $\0$s
and consecutive $1$s. 
Indeed, it is known \cite{zbMATH03947643,MR906959,MR1093759} that 
no single right-to-left and no single left-to-right transducer can normalize
all words $u\in\Aldel^k$. 
\end{example}

In Section~\ref{sec:proof-Berstel-adder}, we provide a constructive proof of
Theorem~\ref{thm:Berstel} based on linear algebra and identities involving
Fibonacci numbers. We use
Theorem~\ref{thm:Berstel} to prove Theorem~\ref{theorem:main}.

\section{Proof of Theorem~\ref{theorem:main}}\label{sec:addtion-Fcal-Z}

Taking the Berstel adder $\Bcal$ as a base, we derive another Mealy machine
\[
\Tcal = (Q\cup \{\texttt{start}\}, \texttt{start}, 
         \Aldel, \Alsig\cup\{\varepsilon\}, \delta_\Tcal, \eta_\Tcal, \phi_\Tcal)
\]
by adding a new initial state $\texttt{start}$
and extending the maps $\delta_\Bcal$ and $\eta_\Bcal$ in the following way:
\begin{align*}
    \delta_\Tcal(q, a) &= 
        \begin{cases}
            \delta_\Bcal(q, a), & \text{ if } q\in Q \text{ and } a\in\Aldel;\\
            \0\0\0.0, &           \text{ if } q=\texttt{start} \text{ and } a=\0;\\
            \1\0\1.7, &           \text{ if } q=\texttt{start} \text{ and } a=\1;\\
            \1\0\0.6, &           \text{ if } q=\texttt{start} \text{ and } a=\2,\\
        \end{cases}\\
        \eta_\Tcal(q, a) &= 
        \begin{cases}
            \eta_\Bcal(q, a), & \text{ if } q\in Q \text{ and } a\in\Aldel;\\
            \varepsilon, &      \text{ if } q=\texttt{start} \text{ and } a\in\Aldel,
        \end{cases}\\
        \phi_\Tcal(q) &= 
        \begin{cases}
            \phi_\Bcal(q), & \text{ if } q\in Q;\\
            \0\0\0,        & \text{ if } q=\texttt{start}.
        \end{cases}
\end{align*}
The Mealy machine $\Tcal$ is illustrated in Figure~\ref{fig:berstel_transducer}
with solid and dashed edges.

Before proving the theorem, we need two lemmas.
The first lemma shows that the first letter of the output word depends only on
the first letter of the input word.

\begin{lemma}\label{lem:properties_of_T}
	Let $v\in \Aldel^*$. Then
	\begin{enumerate}[(i)]
		\item $\Tcal(\0 v)\cdot\Tcal_{\downarrow}(\0v) \in \0\Alsig^*$,
		\item $\Tcal(a v)\cdot\Tcal_{\downarrow}(av)\in \1\Alsig^*$ for every $a\in\{\1,\2\}$.
	\end{enumerate}
\end{lemma}

\begin{proof}
    The proof is based on Figure~\ref{fig:berstel_transducer}. We observe that
    \begin{enumerate}[(i)]
    \item $\phi_\Tcal(\0\0\0.0)$ starts with $\0$ and $\eta_\Tcal(\0\0\0.0, b) = 0$ for all 
    $b \in\{\0,\1,\2\}$,
    \item $\phi_\Tcal(\1\0\1.7)$ and $\phi_\Tcal(\1\0\0.6)$ both start with $\1$, and 
    $\eta_\Tcal(\1\0\1.7, b) = \eta_\Tcal(\1\0\0.6, b) = \1$ for all $b\in\{\0,\1,\2\}$.\qedhere
    \end{enumerate}
\end{proof}

Since the strongly connected component of the Mealy machine $\Tcal$ corresponds
to the Berstel adder $\Bcal$, we have the following relations between the two machines.

\begin{lemma}\label{lem:relation_T_B}
	Let $v\in \Aldel^*$. Then
	\begin{enumerate}[(i)]
		\item $\Bcal(\0v)\cdot\Bcal_{\downarrow}(\0v) = \0\Tcal(\0v)\cdot\Tcal_{\downarrow}(\0v)$,
		\item $\Bcal(\1\0\1v)\cdot\Bcal_{\downarrow}(\1\0\1v) = \0\0\0\Tcal(\1v)\cdot\Tcal_{\downarrow}(\1v)$,
		\item $\Bcal(\2\0\2v)\cdot\Bcal_{\downarrow}(\2\0\2v) = \0\0\1\Tcal(\2v)\cdot\Tcal_{\downarrow}(\2v)$.
	\end{enumerate}
\end{lemma}

\begin{proof}
    In this proof, it is appropriate to introduce the notation
        $M^{s} = (S, s, A, B, \delta, \eta, \phi)$
    to denote the Mealy machine with a chosen initial state $s\in S$.
    In particular, $M^{s}$ satisfies the condition that
    \begin{align}
        M^s(uv) &= M^s(u)\cdot M^t(v)   \label{eq:mealy-split0}\\
        M^s_{\last}(uv) &= M^t_{\last}(v) \label{eq:mealy-split1}
    \end{align}
    for every $u,v\in A^*$ where $t=M^s_{\last}(u)$.

    Moreover, since the strongly connected component of the Mealy machine
    $\Tcal$ corresponds to the Berstel adder $\Bcal$ and both machines
    have the same function $\phi_\Tcal = \phi_\Bcal$, the following equations hold
    for every $q\in Q$ and every $u\in\Aldel^*$:
    \begin{equation}\label{eq:B-to-T}
        \Bcal^{q}(u) = \Tcal^{q}(u),
        \qquad
        \Bcal^{q}_{\last}(u)= \Tcal^{q}_{\last}(u)
        \qquad
        \text{ and }
        \qquad
        \Bcal^{q}_{\downarrow}(u)= \Tcal^{q}_{\downarrow}(u).
    \end{equation}

    Using Equations
\eqref{eq:mealy-split0},
\eqref{eq:mealy-split1} and
\eqref{eq:B-to-T} and the particular form of the machines $\Bcal$ and $\Tcal$, we have
    \begin{align*}
		\Bcal(\0v)\cdot\Bcal_{\downarrow}(\0v)
        &= \Bcal(\0)\cdot\Bcal^{\texttt{000}.0}(v)\cdot\Bcal^{\texttt{000}.0}_{\downarrow}(v)\\
        &= \0\cdot\Tcal^{\texttt{000}.0}(v)\cdot\Tcal^{\texttt{000}.0}_{\downarrow}(v)\\
        &= \0\cdot\Tcal(\0v)\cdot\Tcal_{\downarrow}(\0v).
    \end{align*}
    Also, we have
    \begin{align*}
		\Bcal(\1\0\1v)\cdot\Bcal_{\downarrow}(\1\0\1v)
        &= \Bcal(\1\0\1)\cdot\Bcal^{\texttt{101}.7}(v)\cdot\Bcal^{\texttt{101}.7}_{\downarrow}(v)\\
        &= \0\0\0\cdot\Tcal^{\texttt{101}.7}(v)\cdot\Tcal^{\texttt{101}.7}_{\downarrow}(v)\\
        &= \0\0\0\cdot\Tcal(\1v)\cdot\Tcal_{\downarrow}(\1v)
    \end{align*}
    and similarly,
    \begin{align*}
		\Bcal(\2\0\2v)\cdot\Bcal_{\downarrow}(\2\0\2v)
        &= \Bcal(\2\0\2)\cdot\Bcal^{\texttt{100}.6}(v)\cdot\Bcal^{\texttt{100}.6}_{\downarrow}(v)\\
        &= \0\0\1\cdot\Tcal^{\texttt{100}.6}(v)\cdot\Tcal^{\texttt{100}.6}_{\downarrow}(v)\\
        &= \0\0\1\cdot\Tcal(\2v)\cdot\Tcal_{\downarrow}(\2v).\qedhere
    \end{align*}
\end{proof}

We can now prove the main result.

\begin{proof}[Proof of Theorem~\ref{theorem:main}]
    Let $u\in\Aldel^+$.
    We split the proof into three cases according to the first letter of $u$.
    Let $a\in\Aldel$ and $v\in\Aldel^*$ be such that $u=av$.
    Assume that $|v|=k$. 
    
    If $a=\0$, we have
    \begin{align*}
        \valFc(\0v) 
        &= \valF(\0v)            &\text{(Equation~\eqref{eq:valFc-vs-valF})} \\
        &= \valF(\Bcal(\0v)\cdot\Bcal_{\downarrow}(\0v))
                                    &\text{(Theorem~\ref{thm:Berstel})}\\
        &= \valF(\0\Tcal(\0v)\cdot\Tcal_{\downarrow}(\0v))
                                    &\text{(Lemma~\ref{lem:relation_T_B})}\\
        &= \valF(\Tcal(\0v)\cdot\Tcal_{\downarrow}(\0v))\\
        &= \valFc(\Tcal(\0v)\cdot\Tcal_{\downarrow}(\0v))
                                    &\text{(Equation~\eqref{eq:valFc-vs-valF}, 
                                    Lemma~\ref{lem:properties_of_T}).}
    \end{align*}
    
    If $a=\1$, we have
    \begin{align*}
        \valFc(\1v) 
        &= \valFc(\1\0\1v)          &\text{(Lemma~\ref{lem:neutral_prefix_2}})\\
        &= \valF(\1\0\1v) - F_{k+3} &\text{(Equation~\eqref{eq:valFc-vs-valF})} \\
        &= \valF(\Bcal(\1\0\1v)\cdot\Bcal_{\downarrow}(\1\0\1v)) - F_{k+3}
                                    &\text{(Theorem~\ref{thm:Berstel})}\\
        &= \valF(\0\0\0\cdot\Tcal(\1v)\cdot\Tcal_{\downarrow}(\1v)) - F_{k+3} 
                                    &\text{(Lemma~\ref{lem:relation_T_B})}\\
        &= \valF(\Tcal(\1v)\cdot\Tcal_{\downarrow}(\1v)) - F_{k+3} \\
        &= \valFc(\Tcal(\1v)\cdot\Tcal_{\downarrow}(\1v))
                                    &\text{(Equation~\eqref{eq:valFc-vs-valF}, 
                                    Lemma~\ref{lem:properties_of_T}).}
    \end{align*}

    If $a=\2$, we proceed similarly to the previous case and we have
    \begin{align*}
        \valFc(\2v) 
        &= \valFc(\2\0\2v)           &\text{(Lemma~\ref{lem:neutral_prefix_2}})\\
        &= \valF(\2\0\2v) - 2F_{k+3} &\text{(Equation~\eqref{eq:valFc-vs-valF})} \\
        &= \valF(\Bcal(\2\0\2v)\cdot\Bcal_{\downarrow}(\2\0\2v)) - 2F_{k+3} 
                                    &\text{(Theorem~\ref{thm:Berstel})}\\
        &= \valF(\0\0\1\cdot\Tcal(\2v)\cdot\Tcal_{\downarrow}(\2v)) - 2F_{k+3}
                                    &\text{(Lemma~\ref{lem:relation_T_B})}\\
        &= \valF(\Tcal(\2v)\cdot\Tcal_{\downarrow}(\2v)) + F_{k+3} - 2F_{k+3} \\
        &= \valF(\Tcal(\2v)\cdot\Tcal_{\downarrow}(\2v)) - F_{k+3}\\
        &= \valFc(\Tcal(\2v)\cdot\Tcal_{\downarrow}(\2v))
                                    &\text{(Equation~\eqref{eq:valFc-vs-valF}, 
                                    Lemma~\ref{lem:properties_of_T}).\qedhere}
    \end{align*}
\end{proof}

In analogy with Example~\ref{example:berstel_addition}, we
illustrate Theorem~\ref{theorem:main} in the following example.
See also the behavior of $\Tcal$ on all ternary words of length $\leq 3$ in 
Table~\ref{table:Berstel-adder-length-3-word}.

\begin{example}\label{example:Tcal_addition}
	We feed the modified Berstel adder $\Tcal$ in Figure~\ref{fig:berstel_transducer} with the same word $u=\2\2\2\0\1\2\1$ as in Example~\ref{example:berstel_addition},
	this time obtaining
	\[
	\texttt{start} \xrightarrow{\2/\varepsilon} \1\0\0.6\xrightarrow{\2/\1} \1\0\0.5 \xrightarrow{\2/\1}
	\0\1\0.4 \xrightarrow{\0/\0} \1\0\1.6 \xrightarrow{\1/\1} \0\1\0.4
	\xrightarrow{\2/\1} \0\0\1.2 \xrightarrow{\1/\0} \1\0\0.5;
	\]
	therefore the last state is $\Tcal_{\last}(u)=\1\0\0.5$ and we obtain $\Tcal(u)\cdot \Tcal_{\downarrow}(u) = \1\1\0\1\1\0 \cdot \1\0\0$. Interpreting
	the results in the Fibonacci complement numeration system, we observe 
	\[
		\valFc(\2\2\2\0\1\2\1) = (-42)+42+16+3+4+1= 24 = (-34)+34+13+8+3= \valFc(\1\1\0\1\1\0\1\0\0).
	\]
\end{example}

\begin{example}\label{example:Tcal_addition-2}
Here is how the Mealy machine $\Tcal$ can be used to compute the sum $-1+(-9)$.
First, we express $(-1)$ and $(-9)$ by their Fibonacci complement representation
and we pad the shorter word with an appropriate neutral prefix ($\0\0$'s if it
starts with $\0$ or $\1\0$'s if it starts with $\1$, see
Definition~\ref{def:neutral-prefix}) so that they have the same length. Then,
we add them digit by digit to obtain a ternary word in $\{\0,\1,\2\}^*$:
\begin{equation*}
    \begin{aligned}
        &-1  \quad&\texttt{1010101} & \\
        +&(-9)  \quad&\texttt{1000101} & \\[-.4cm]
        \cline{1-3}
        &-10  \quad&\texttt{2010202} &
    \end{aligned}
\end{equation*}
Note that the resulting word $u=\texttt{2010202}$ coincides with the one
in Example~\ref{example:berstel_addition-2},
where we show properties of addition in Fibonacci numeration system.
Reading from left to right and giving the word $u=\texttt{2010202}$ as input to
the Mealy machine $\Tcal$ (see Figure~\ref{fig:berstel_transducer}), we obtain
the following path from the initial state $\texttt{start}$:
\[
    \texttt{start} \xrightarrow{\2/\varepsilon}  \texttt{100}.6
                   \xrightarrow{\0/\1}  \texttt{001}.1
                   \xrightarrow{\1/\0}  \texttt{010}.4
                   \xrightarrow{\0/\0}  \texttt{101}.6
                   \xrightarrow{\2/\1}  \texttt{100}.6
                   \xrightarrow{\0/\1}  \texttt{001}.1
                   \xrightarrow{\2/\0}  \texttt{100}.6.
\]
Therefore, the output word is $\Tcal(u)=\texttt{100110}$
and the path ends in state $\Tcal_{\last}(u)=\texttt{100}.6$.
Removing the last digit (6) of the last state,
we obtain the three-letter extra output word $\Tcal_{\downarrow}(u)=\texttt{100}$.
Concatenating the two words gives
$\texttt{100110}\cdot\texttt{100}$, which is a Fibonacci complement
representation of the sum $-1+(-9)$. We confirm that its Fibonacci
complement value is correct:
\[
    \valFc\left(\Tcal(u)\cdot\Tcal_{\downarrow}(u)\right)
    = \valFc\left(\texttt{100110}\cdot\texttt{100}\right)
    = 3+8+13-34=-10.
\]
\end{example}

\section{A characterization by increasing bijection}\label{sec:increasing-bijection}

The goal of this section is to show that the numeration system $\repFc \colon
\Z\to D$ is characterized 
by the fact of being an
increasing map with respect to some total order on its codomain
$D = \Sigma(\Sigma\Sigma)^*\setminus
\left( \Sigma^*\1\1\Sigma^* \cup \0\0\0\Sigma^* \cup \1\0\1\Sigma^*\right)$
where $\Sigma=\{\0,\1\}$; see
Theorem~\ref{thm:repF_is_increasing}.
The proof of the theorem is based on similar results for the maps
$\valF$ and $\repF$ which we start with.

\subsection{Ordering in the Fibonacci numeration system for $\N$}

In this short section, we recall known facts about the Fibonacci numeration system
also known as the Zeckendorf numeration system. We recall that $\repF$ is increasing
with respect to the radix order.
For completeness, we also provide the proofs of the results.

\begin{lemma}\label{lem:Fibo-intervals}
    Let $w\in\Sigma^*\setminus (\Sigma^* \1\1 \Sigma^* \cup \0\Sigma^*)$.
    If $w \neq \varepsilon$
    and $k\geq1$ is an integer, then
    \[
        |w| = k \text{ if and only if } F_{k-1}\leq \valF(w) < F_k.
    \]
    Moreover, $|w| = 0 \text{ if and only if } \valF(w) = 0$.
\end{lemma}

\begin{proof}
    Let $w\in\Sigma^*\setminus (\Sigma^* \1\1 \Sigma^* \cup \0\Sigma^*)$.
    If $w = \varepsilon$ then $\valF(w) = 0$.
    If $w \neq \varepsilon$ then $w\in\Sigma^* \1 \Sigma^*$ and $\valF(w) > 0$.
    Let $k \in \{1,2\}$. Then $|w| = k$ if and only if $w=p_k$
    where $p_k$ is the prefix of length $k$ of the word $\1\0$.
    Let $n\in\N$.
    Then $F_{k-1} \leq n < F_k $ if and only if $n=k=\valF(p_k)$.

    Let $k\geq 3$ and let it hold for all integers $\ell$ up to $k-1$
    that $|w| = \ell$ if and only if $F_{\ell - 1} \leq \valF(w) < F_\ell$.
    Let $w = w_{k-1} w_{k-2} \cdots w_0 \in \Sigma^*\setminus (\Sigma^* \1\1 \Sigma^* \cup \0\Sigma^*)$.
    Then applying induction hypothesis on the suffix
    $w_{q-1} \cdots w_1 w_0 \in \Sigma^*\setminus (\Sigma^* \1\1 \Sigma^* \cup \0\Sigma^*)$
    such that $w = \1\0^{k-q-1} w_{q-1}\cdots w_0$
    it holds that $F_{q-1} \leq \sum_{i=0}^{q-1} w_i F_i < F_{q}$.
    As $w_{k-1} = \1$ and $q \leq k-2$, we have
    \[
         F_{k-1} \leq \valF(w) = \sum_{i=0}^{k-1} w_i F_i = F_{k-1} + \sum_{i=0}^{q-1} w_i F_i
        < F_{k-1} + F_q \leq F_{k-1} + F_{k-2} = F_k.\qedhere
    \]
\end{proof}

\begin{lemma}\label{lem:Fibo-bijection}
    The map $\valF\colon \Sigma^*\setminus (\Sigma^* \1\1 \Sigma^* \cup \0\Sigma^*)\to \N$
    is a bijection.
\end{lemma}

\begin{proof}
    (Injectivity): 
    Let $u,v\in\Sigma^*\setminus (\Sigma^* \1\1 \Sigma^* \cup \0\Sigma^*)$ be such that
    $\valF(u) = \valF(v)$. Assume by contradiction that $u \neq v$.
    By Lemma~\ref{lem:Fibo-intervals}, $|u| = |v|$.
    Let $p\in\Sigma^*$
    be the longest common prefix of $u$ and $v$. 
    Therefore there exist $\ell\in\N$ and
    $t,s \in \Sigma^\ell$ such that without loss of generality
    $u = p \0 s$ and $v = p \1 t$.
    Then by Lemma~\ref{lem:Fibo-intervals}, 
    \[
        \valF(v) - \valF(u) = \valF(\1t) - \valF(\0s) 
        > F_{\ell} - F_{\ell-1} = F_{\ell - 2}
        \geq 0,
    \]
    which is a contradiction.

    (Surjectivity): 
    Let $S = \Sigma^*\setminus (\Sigma^* \1\1 \Sigma^* \cup \0\Sigma^*)$.
    We proceed by induction on $n\in\N$.
    If $n = 0$ then $w = \varepsilon\in S$ fulfills  $\valF(w) = n$.
    Let $n \geq 1$. Induction hypothesis: for every $m\in\N$
    such that $m < n$ there exists 
    $w \in S$
    such that $\valF(w) = m$.
    There exists a unique $k\in\N$ such that $F_{k-1} \leq n < F_k$.
    We have $0 \leq n - F_{k-1} < F_k - F_{k-1} = F_{k-2}$.
    By Lemma~\ref{lem:Fibo-intervals}, $|w| \leq k-2$.
    By induction hypothesis, there exists 
    $w \in S$
    such that $\valF(w) = n - F_{k-1}$.
    Let $\ell = k - 1 - |w|$ and $u = \1\0^\ell w$.
    Clearly $\ell \geq 1$. Thus $u\in S$ and
    \[
        \valF(u) = \valF(\1\0^\ell w) = F_{k-1} + \valF(w)
        = F_{k-1} + n - F_{k-1} = n.\qedhere
    \]
\end{proof}

The map $\repF$ is defined as the inverse map $\valF^{-1}$.
It corresponds to the Zeckendorf numeration system \cite{MR308032}.

We can now show that $\valF$ is an increasing map.
Recall that the radix order is a total order on $\Sigma^*$ such that,
for every $u,v\in\Sigma^*$,
$u <_{\rad}
v$ if and only if $|u| < |v|$ or $|u| = |v|$ and $u <_{lex} v$.

\begin{lemma}\label{lem:valF_is_increasing}
    Let $S= \Sigma^*\setminus
	\left( \Sigma^*\1\1\Sigma^* \cup \0\Sigma^*\right)$.
	The map $\valF$
    is an increasing bijection from 
    $(S,<_{\rad})$ to $(\N,<)$.
\end{lemma}

\begin{proof}
    It follows from Lemma~\ref{lem:Fibo-bijection} that 
    $\valF$ is a bijection. We show that it is increasing.

    Let $u,v \in \Sigma^*\setminus
	\left( \Sigma^*\1\1\Sigma^* \cup \0\Sigma^*\right)$ be such that $u <_{\rad} v$.
    If $|u| < |v|$ then by Lemma~\ref{lem:Fibo-intervals},
    $\valF(u) < \valF(v)$. Assume that $|u| = |v|$.
    Then $u <_{lex} v$.
    Then there exists
    $p\in \Sigma^*$ such that
    $u = p\0 s$, $v = p\1 t$ for some
    $t,s\in \Sigma^*\setminus \Sigma^* \1\1 \Sigma^*$.
    Denote $\ell = |t|=|s|$. 
    Thus from Lemma~\ref{lem:Fibo-intervals},
    \[
        \valF(v) - \valF(u) 
        =\valF(p\1 t) - \valF(p\0 s) 
        = \valF(\1 t) - \valF(\0 s) >
        F_\ell - F_{\ell}
        = 0. \qedhere
    \]
\end{proof}

\subsection{Proof of Proposition~\ref{prop:Frep}}

The following lemma
allows us to determine whether $\valFc(w)$ 
is nonnegative or negative based only on the first digit of $w$. 

\begin{lemma}\label{lem:intervals-of-values}
    For every word
    $w\in\Sigma^{+}\setminus \Sigma^*\1\1\Sigma^*$
    such that $|w|=k$.
    We have
    \begin{enumerate}
        \item $w\in\0\Sigma^*$ if and only if $0\leq\valFc(w)<F_{k-1}$,
        \item $w\in\1\Sigma^*$ if and only if $-F_{k-2}\leq\valFc(w)< 0$.
    \end{enumerate}
\end{lemma}
\begin{proof}
	Let 
	$w = w_{k-1} w_{k-2} \cdots w_0 \in \Sigma^{+}\setminus \Sigma^*\1\1\Sigma^*$
    for some $k\geq 1$.
    First we prove the implication from left to right for both cases.
    \begin{enumerate}
    \item Let $\ell \leq k-2$ be the unique integer such that
        $w = v\1 s$ for some $v\in\0^+$ and $s\in\Sigma^{\ell}$.
        Then $0\leq \valFc(w) = \valF(w) = \valF(\1 s) < F_{\ell+1} \leq F_{k-1}$
        by Lemma~\ref{lem:Fibo-intervals}.
    \item 
        Using Equation~\eqref{eq:valFc-vs-valF}
        and part (1), we have
        $-F_{k-2}\leq -F_{k-2} + \sum_{i=0}^{k-2}w_i F_i 
        = \valFc(w) 
        = \valF(w)-F_{k} 
        = \valFc(\0w)-F_{k}
        < F_{k} - F_{k} = 0$.
    \end{enumerate}
    The converse follows from the following observation. 
    Let $X = X_1 \cup X_2$ be such that $X_1 \cap X_2 = \emptyset$.
    Let $F\colon X\to Y$
    be a map, such that $F(X_1) \cap F(X_2) = \emptyset$.
    Then it holds that, for every $x\in X$, $F(x)\in F(X_1)$ implies $x\in X_1$.
    Indeed, if $x\in X_2$, then $F(x) \in F(X_1) \cap F(X_2)$, a contradiction.
    It suffices to use this observation with $F = \valF$, $X_1 = \0\Sigma^*$,
    $X_2 = \1\Sigma^*$.
\end{proof}

Recall that $D = \Sigma(\Sigma\Sigma)^*\setminus
\left( \Sigma^*\1\1\Sigma^* \cup \0\0\0\Sigma^* \cup \1\0\1\Sigma^*\right)$
is the set of odd-length words $w\in D$ with no consecutive $\1$s and which
are not starting with a neutral prefix.
The following lemma strengthens Lemma~\ref{lem:intervals-of-values}
for $w\in D$.

\begin{lemma}\label{lem:excercise_on_Fibo}
    Let $w\in D$ and $k \in\N$. 
    \begin{enumerate}
        \item $w\in\0\Sigma^{2k}$
            if and only if $F_{2k-2}\leq\valFc(w)<F_{2k}$.
        \item $w\in\1\Sigma^{2k}$
            if and only if $-F_{2k-1}\leq\valFc(w)<-F_{2k-3}$.
        \item Also, $w=\1$ if and only if $\valFc(w)=-1$.
    \end{enumerate}
\end{lemma}

\begin{proof}
    Let $w\in D$ and $k \in\N$. 

    (1): Let $w\in\0\Sigma^*$ and $|w| = 2k+1$. Then $w =v \1 s$ for some $v\in\{\0,\0\0\}$
    and $s\in \Sigma^*$.
    Applying Equation~\eqref{eq:valFc-vs-valF} 
    we have 
    $\valFc(w)
    = \valF(v\1 s)
    = \valF(\1 s)$. 
    If $v = \0$ then $|s| = 2k-1$ and 
    from Lemma~\ref{lem:Fibo-intervals},
    $F_{2k-1} \leq  \valF(\1 s)  < F_{2k}$.
    If $v = \0\0$ then $|s| = 2k-2$ and 
    from Lemma~\ref{lem:Fibo-intervals},
    $F_{2k-2} \leq  \valF(\1 s)  < F_{2k-1}$.
    Together, $F_{2k-2}\leq\valFc(w)<F_{2k}$.

    On the other hand, let $w\in D$ be such that
    $F_{2k-2} \leq\valFc(w)< F_{2k}$. 
    By Lemma~\ref{lem:intervals-of-values},
    $w\in\0 \Sigma^{2\ell}$ for some $\ell \leq k$.
    From the first part of this proof, we have $k = \ell$ and $w \in \0\Sigma^{2k}$.

    (2): Let $w\in\1\Sigma^*$ and $|w| = 2k+1$.
    From Lemma~\ref{lem:Fibo-intervals},
    $F_{2k}\leq \valF(w)$.
    From Equation~\eqref{eq:valFc-vs-valF},
    \[
        \valFc(w) 
        = \valF(w) - F_{2k+1}
        \geq F_{2k} - F_{2k+1}
        = -F_{2k-1}.
    \]
    Let $s\in\Sigma^*$ be the suffix of length $\ell \in\N$ 
    such that $w=\1v \1 s$ for some $v\in\0^+$.
    Clearly $\ell \leq 2k-3$.
    From Lemma~\ref{lem:Fibo-intervals},
    $\valF(\1 s) < F_{\ell+1}$.
    Thus
    \[
        \valFc(w) 
        = \valFc(\1v\1 s) = \valF(\0 v\1 s)-F_{2k-1}
        < F_{\ell +1}- F_{2k-1} \leq F_{2k-2}- F_{2k-1} = -F_{2k-3}.
    \]

    On the other hand, let $w\in D$ be such that
    $-F_{2k-1} \leq\valFc(w)< -F_{2k-3}$. By Lemma~\ref{lem:intervals-of-values},
    $w\in\1 \Sigma^{\ell}$
    for some $\ell \leq 2k$.
    From the first part of this proof, we have $\ell = 2k$ and $w \in \1\Sigma^{2k}$.

    (3): If $w = \1$ then $\valFc(w) = -1$.
    The converse follows from previous parts (1) and (2).
\end{proof}

We can now prove that $\valFc$ is a bijection $D\to\Z$.

\begin{proof}[Proof of Proposition~\ref{prop:Frep}]
    (Injectivity): Let $u,v\in D$
    be such that $\valFc(u) = \valFc(v)$.
    By Lemma~\ref{lem:excercise_on_Fibo}, 
    $u$ and $v$ have the same length $k\in\N$
    and the same first digit $\mathtt{d}\in\{\0,\1\}$. 
    Using Equation~\eqref{eq:valFc-vs-valF},
    $\valF(u) = \valFc(u) + \mathtt{d} F_k = \valFc(v) + \mathtt{d} F_k = \valF(v)$.
    By Lemma~\ref{lem:Fibo-bijection}, $u = v$.

    (Surjectivity): It holds that
    $\Z = \{-1\}\cup \bigcup_{k = 0}^{+\infty} [F_{2k-2}, F_{2k})
    \cup [-F_{2k-1}, -F_{2k-3})$,
    a disjoint union.
    If $n = -1$ then $w = \1$ fulfills $\valFc(w) = n$.
    Let $n\in\Z\setminus\{-1\}$. 

    Assume $n \geq 0$. Then there exists a unique $k\in\N$ 
    such that $F_{2k-2} \leq n < F_{2k}$. 
    From Lemma~\ref{lem:Fibo-bijection}, there exists a unique $w\in
    \Sigma^*\setminus
    (\Sigma^*\1\1\Sigma^* \cup \0\Sigma^*)$ such that $\valF(w) = n$.
    If $n < F_{2k-1}$, then by Lemma~\ref{lem:Fibo-intervals},
    $|w| = 2k-1$ and the word $\0\0 w \in D$
    fulfills $\valFc(\0\0w) = \valF(\0\0w) = \valF(w) = n$.
    If $F_{2k-1} \leq n$, then by Lemma~\ref{lem:Fibo-intervals},
    $|w| = 2k$ and the word $\0w \in D$
    fulfills $\valFc(\0w) = \valF(\0w) = \valF(w) = n$.

    Assume $n < -1$. Then there exists a unique integer $k\geq1$ 
    such that $-F_{2k-1} \leq n < -F_{2k-3}$.
    From Lemma~\ref{lem:Fibo-bijection},
    there exists $w\in\Sigma^*\setminus
    (\Sigma^*\1\1\Sigma^* \cup \0\Sigma^*)$ such that $\valF(w) =F_{2k-1} + n$.
    Let $\ell = 2k - |w|$. Let $u = \1\0^\ell w$. Then $|u| = 2k + 1$ and thus
    \[
        \begin{aligned}
        \valFc(u) 
            &= -F_{2k+1} + \valF(\1\0^\ell w) \\
            &= -F_{2k+1} +F_{2k} + \valF(w) \\
            &= -F_{2k+1} +F_{2k} + F_{2k-1} + n = n. \qedhere
        \end{aligned}
    \]
\end{proof}

\subsection{Ordering in the Fibonacci complement numeration system}

We define the reversed-radix order as a total order such that $u <_{\rev} v$ if and only if
    $|u| > |v|$ or $|u| = |v|$ and $u <_{lex} v$.
We define a total order on $\Sigma^*$ as follows.

\begin{definition}[total order $\prec$]\label{def:order}
    For every $u,v \in \Sigma^*$,
	we define $u \prec v$ if and only if
	\begin{itemize}
		\item $u\in \1 \Sigma^*$ and $v\in \0 \Sigma^*$, or
        \item $u,v \in \0 \Sigma^*$ and $u <_{\rad} v$, or
        \item $u,v \in \1 \Sigma^*$ and $u <_{\rev} v$.
	\end{itemize}
\end{definition}

The map $\valFc$ is an increasing bijection with respect to this total order on
$D$.

\begin{lemma}\label{lem:valFc_is_increasing}
	The map $\valFc$
    is an increasing bijection from $(D,\prec)$ to $(\Z,<)$.
\end{lemma}

\begin{proof}
	Let $u,v\in D$ 
    be such that $u\prec v$.
    Let $k,\ell\in\N$ be such that $|u|=2k+1$ and $|v|=2\ell+1$.
    \begin{itemize}
        \item
    If $u\in \1 \Sigma^*$ and $v\in \0 \Sigma^*$, then
    from Lemma~\ref{lem:intervals-of-values} 
    $\valFc(u)<0\leq\valFc(v)$.
        \item
    Assume that $u,v \in \0 \Sigma^*$ and $|u| < |v|$. 
    Thus $k\leq \ell-1$.
    From Lemma~\ref{lem:excercise_on_Fibo}, we have
    $\valFc(u)<F_{2k}\leq F_{2(\ell-1)}\leq\valFc(v)$.
        \item
    Assume that $u,v \in \1 \Sigma^*$ and $|u| > |v|$. 
    Thus $k-1\geq \ell$.
    From Lemma~\ref{lem:excercise_on_Fibo}, we have
    $\valFc(u)<-F_{2k-3}=-F_{2(k-1)-1}\leq -F_{2\ell-1}\leq\valFc(v)$.
    \item Assume that $u,v \in \mathtt{d}\Sigma^*$ 
        for some $\mathtt{d}\in\{\0,\1\}$ and $|u| = |v|$. 
        In this case, we have $u<_{lex}v$.
        From Lemma~\ref{lem:valF_is_increasing}, 
        $\valF(u)<\valF(v)$.
        Since $u$ and $v$ start with the same digit we have 
        $\valFc(v)-\valFc(u)=\valF(v)-\valF(u)>0$.
        Thus $\valFc(u)<\valFc(v)$.\qedhere
    \end{itemize}
\end{proof}

We can now prove that $\repFc$ is characterized by the fact of being
an increasing bijection.

\begin{maintheorem}\label{thm:repF_is_increasing}
	Let $f:\Z\to D$ where $D=\Sigma(\Sigma\Sigma)^*\setminus
	\left( \Sigma^*\1\1\Sigma^* \cup \0\0\0\Sigma^* \cup \1\0\1\Sigma^*\right)$.
    The map $f$
    is an increasing bijection from 
    $(\Z,<)$ to $(D,\prec)$ such that $f(0)=\0$
    if and only if $f=\repFc$.
\end{maintheorem}

\begin{proof}
    Let $D= \Sigma(\Sigma\Sigma)^*\setminus
	\left( \Sigma^*\1\1\Sigma^* \cup\0\0\0\Sigma^* \cup \1\0\1\Sigma^*\right)$.
    The map $\repFc\colon\Z\to D$ is the inverse of the map $\valFc$,
    which by Lemma~\ref{lem:valFc_is_increasing} is an increasing bijection $D \to \Z$
    with respect to the order $\prec$. Hence $\repFc\colon \Z\to D$
    is an increasing bijection with respect to the order~$\prec$.
    Moreover $\repFc(0)=\0$.

    Let $f\colon \Z\to D$ be an increasing bijection with respect to the order~$\prec$
    such that $f(0) = \0$.
    There is a unique increasing bijection 
    $(\Z,<)$ to $(D,\prec)$ 
    such that $0\mapsto \0$.
    Thus $f=\repFc$.
\end{proof}

Theorem~\ref{thm:repF_is_increasing} allows us to prove that $\repFc$, the
Fibonacci analogue of the two's complement, is the same as the Dumont-Thomas
numeration system associated to one of the two biinfinite
periodic points of the Fibonacci substitution
\cite{2302.14481}.

\section{A constructive proof of Theorem~\ref{thm:Berstel}}\label{sec:proof-Berstel-adder}

The Berstel adder $\Bcal$ was given in \cite{zbMATH03947643} without proof.
A proof of Theorem~\ref{thm:Berstel} was provided later in \cite[Corollary 4]{MR1705857}
based on the numeration system in the real base $\tau=\frac{1+\sqrt{5}}{2}$,
see also \cite{MR942576,MR1093759}.
Another proof follows from \cite[\S 2.3.2.3]{MR2742574} where it is proved that
normalization in the real base $\beta$ is computable by a finite automaton when
$\beta$ is a Pisot number.
Alternatively, the correctness of the Berstel adder can be proved with a
decision procedure; see \cite[Remark 2.1, p.~41]{MR3518158}.

The goal of this section is to provide a new proof of 
Theorem~\ref{thm:Berstel} based solely on linear algebra and
identities involving Fibonacci numbers. 
The proof is constructive. More precisely,
we show that there is a canonical way
of mapping a ternary word to a binary word having the same Fibonacci value,
provided the choice has been made for shorter words; see Lemma~\ref{lem:tree-structure}
and Proposition~\ref{prop:uniqueness-new}.
Then we introduce an equivalence
relation on ternary words (see Definition~\ref{def:equivalence}),
describing the behavior of the translation to binary words.
We show that there are finitely many equivalence classes,
which proves that the computation can be performed by a finite-state transducer.

First, we prove two identities on Fibonacci numbers.
These identities are not part of lists of well-known identities for Fibonacci
numbers
\cite{MR0245499,zbMATH03316160,MR1954396,MR1397498,MR3077152}.
Also, we discovered the second identity thanks to the existence of the sequence
\[
        0, 5, 39, 272, 1869, 12815, 87840, 602069, 4126647, 28284464, \dots
\]
referenced as sequence \href{https://oeis.org/A003482}{A003482} in OEIS
\cite{OEISA003482}, which was yet to be associated with
$(\sum_{i=0}^{2k-1}F_i^2)_{k\geq0}$.
Recall that, following the literature on numeration systems, 
we define the Fibonacci sequence
$(F_n)_{n\geq 0}$ 
with the recurrence relation
$F_{n} = F_{n-1} + F_{n-2}$, for all $n \geq 2$, and
with the initial conditions
$ F_0 = 1$ and $F_1 = 2$.

\begin{lemma}\label{lem:Fibo-identities}
    The following identities involving Fibonacci numbers hold for $k\geq0$:
    \begin{enumerate}[(i)]
        \item \label{item:fiboidentity1} $\sum_{i=0}^{2k-1}(-1)^iF_iF_{2k-i}=-F_{2k-2}$,
        \item \label{item:fiboidentity3} $\sum_{i=0}^{2k-1}F_i^2 = F_{2k-2}F_{2k+1}$.
    \end{enumerate}
\end{lemma}

\begin{proof}
    \eqref{item:fiboidentity1}.
    When $k=0$, we have
        $\sum_{i=0}^{2k-1}(-1)^iF_iF_{2k-i}=0=-F_{-2}$.
        We prove the identity by induction and use
        the other identity $F_mF_n+F_{m-1}F_{n-1}=F_{m+n+1}$
        \cite[\S 1]{MR1954396}.
        If $k\geq1$, we have
        \begin{align*}
        \sum_{i=0}^{2k-1}(-1)^iF_iF_{2k-i}
            &= F_0F_{2k} + \sum_{i=1}^{2k-2}(-1)^iF_iF_{2k-i} -F_{2k-1}F_{1}\\
            &= F_{2k} + \sum_{i=1}^{2k-2}(-1)^i\left(-F_{i-1}F_{2k-i-1}+F_{2k+1}\right) -2F_{2k-1}\\
            &= F_{2k} + \sum_{j=0}^{2(k-1)-1}(-1)^{j}F_{j}F_{2(k-1)-j}+\sum_{i=1}^{2k-2}(-1)^iF_{2k+1} -2F_{2k-1}\\
            &= F_{2k} + (-F_{2(k-1)-2}) + 0 -2F_{2k-1}\\
            &= (F_{2k-1} + F_{2k-2}) - F_{2k-4} -2F_{2k-1}\\
            &= F_{2k-2} - F_{2k-4} -F_{2k-1}\\
            &= -F_{2k-3} - F_{2k-4} = -F_{2k-2}.
        \end{align*}

    \eqref{item:fiboidentity3}.
    The proof uses the Cassini identity
    $F_{n+1}F_{n-1}-F_n^2=(-1)^n$
    and $\sum_{i=-1}^{n}F_i^2 = F_{n}F_{n+1}$ whose geometrical proof is well-known.
    See \cite[\S 1]{MR1954396} where both identities are proven.
    Thus
    \begin{align*}
        \sum_{i=0}^{2k-1}F_i^2 
        &=\sum_{i=-1}^{2k-1}F_i^2 - 1
        = F_{2k-1}F_{2k} - 1\\
        &= (F_{2k-2}+F_{2k-3})F_{2k} - 1\\
        &= F_{2k-2}F_{2k} + F_{2k-3}(2F_{2k-2}+F_{2k-3}) - 1\\
        &= F_{2k-2}(F_{2k}+2F_{2k-3}) + F_{2k-3}^2 - 1\\
        &= F_{2k-2}(F_{2k}+2F_{2k-3}) + F_{2k-2}F_{2k-4} -(-1)^{2k-3} - 1\\
        &= F_{2k-2}(F_{2k}+2F_{2k-3} + F_{2k-4})
        = F_{2k-2}F_{2k+1}.\qedhere
    \end{align*}
\end{proof}

    Let $u,w\in\{\0,\1,\2\}^*$ be two finite ternary sequences of integers.
    If $u$ and $w$ have the same value in base $2$,
    then $u\0$ and $w\0$ have also the same
    value in base $2$, since 
    the values of $u\0$ and  $w\0$ are twice those of $u$ and $w$.
    This is no longer true in the Fibonacci numeration system $\Fcal$.
    Still, the next lemma shows that the difference between the values 
    of the words $u\0$ and $w\0$ is at most 1.
    Its proof is based on linear algebra and the three identities
    involving Fibonacci numbers stated above.

\begin{lemma}\label{lem:append0-change-at-most-1}
    Let $u,w \in \Aldel^*$ be such that $\valF(u)=\valF(w)$. Then
    \begin{enumerate}[(i)]
        \item $\valF(u\0)-\valF(w\0)\in\{-1,0,+1\}$;
        \item if $w\in\Aldel^*\0\0\0$, then
            $\valF(u\0)-\valF(w\0)\in\{0,+1\}$;
        \item if $w\in\Aldel^*\1\0\1$, then
            $\valF(u\0)-\valF(w\0)\in\{-1,0\}$.
    \end{enumerate}
\end{lemma}

\begin{proof}
    (i) The proof is based on linear algebra.
    Up to padding $u$ and $w$ with $\0$'s, we can assume
    without loss of generality that $u$ and $w$ have the same length 
    and that this length is even.
    Let $f=(F_{2k-1},\dots,F_1,F_0)^t\in\R^{2k}$ be a vector of Fibonacci numbers
    where $2k$ is the length of $u$ and $w$.
    We extend $f$ to a set 
    $\{f,g_1,\dots,g_{2k-2},h\}$
    with
    $g_1=(1,-1,-1,0,\dots,0)^t,
         g_2=(0,1,-1,-1,0,\dots,0)^t,
         \dots,
         g_{2k-2}=(0,\dots,0,1,-1,-1)^t \in\R^{2k}$
         and
    $h=(F_0,-F_1,\dots,(-1)^{2k-1}F_{2k-1})^t \in\R^{2k}$.
    The set is a base of $\R^{2k}$ 
    since its vectors are linearly independent.
    Indeed,
    $c_0 f + c_1 g_1 + \ldots + c_{2k-2} g_{2k-2} + c_{2k-1} h = 0$
    implies that
    $(c_0 f + c_1 g_1 + . . . + c_{2k-2} g_{2k-2} + c_{2k-1} h) \cdot f = 0$,
    which holds if and only if
    $c_0 = 0$. Similarly we obtain that $c_{2k-1} = 0$. 
    Also, we prove that $c_i = 0$ for every $1\leq i\leq 2k-2$ by induction on $i$.

    The base is not orthogonal, but $f$ and $h$ are each orthogonal to the other vectors of the base.
    We form the invertible matrix $K=[f,g_1,\dots,g_{2k-2},h]$.
    There exists a unique vector $\alpha=(\alpha_0,\dots,\alpha_{2k-1})^t\in\R^{2k}$
    such that $\vec{u}-\vec{w}=K\alpha$,
    where we consider $\vec{u}=(u_{2k-1},\dots,u_0)^t$ and $\vec{w}=(w_{2k-1},\dots,w_0)^t$ 
    as column vectors in $\R^{2k}$.
    By hypothesis, we have that $f^t\cdot (\vec{u}-\vec{w})=0$.
    Since $f$ is orthogonal to the other vectors of the base, we have
    \[
        0
        = f^t\cdot(\vec{u}-\vec{w})
        = f^t K\alpha
        = f^t f \alpha_0.
    \]
    Therefore $\alpha_0=0$.
    Moreover, $\alpha = K^{-1}(\vec{u}-\vec{w})$.
    By definition of the inverse, the last row of $K^{-1}$ must be orthogonal to
    $f$ and to the $g_i$'s, thus it must be parallel to $h^t$. Since $K^{-1}K$
    is the identity matrix, the last row of $K^{-1}$ must be
    $\frac{1}{h^t\cdot h}h^t$.
    This allows us to express $\alpha_{2k-1}$ as
    $\alpha_{2k-1} =\frac{1}{h^t\cdot h}h^t\cdot (\vec{u}-\vec{w})$.
    Now consider the augmented vector $f'=(F_{2k},F_{2k-1},\dots,F_1,F_0)^t\in\R^{2k+1}$.
    Also let $\alpha'=(\alpha_0,\dots,\alpha_{2k-1},0)^t\in\R^{2k+1}$
    and $K'=\left(\begin{smallmatrix}K&0\\0&0\end{smallmatrix}\right)$ 
        be the matrix $K$ augmented by a zero column vector on the right
        and by a zero row vector at the bottom.
    Using Lemma~\ref{lem:Fibo-identities}~\eqref{item:fiboidentity1}
    and~\eqref{item:fiboidentity3},
    we have
    \begin{equation}\label{eq:valu0-valw0}
    \begin{aligned}
        f'^t\cdot(\vec{u\0}-\vec{w\0})
        &= f'^t\cdot K'\cdot \alpha'
        = f'^t\cdot \left(\begin{smallmatrix}f\\0\end{smallmatrix}\right)\alpha_0
            + 0
            + f'^t\cdot\left(\begin{smallmatrix}h\\0\end{smallmatrix}\right)\alpha_{2k-1}
        = f'^t\left(\begin{smallmatrix}h\\0\end{smallmatrix}\right)\alpha_{2k-1} \\
            &= \frac{h^t\cdot (\vec{u}-\vec{w})}{h^t\cdot h}\sum_{i=0}^{2k-1}(-1)^iF_iF_{2k-i}
             = \frac{h^t\cdot (\vec{u}-\vec{w})}{\sum_{i=0}^{2k-1}F_i^2}(-F_{2k-2})
             = \frac{h^t\cdot (\vec{u}-\vec{w})}{-F_{2k+1}}.
    \end{aligned}
    \end{equation}

    Using the identity $\sum_{i=0}^{2k-1}F_i = F_{2k+1}-2$
    (see \cite[\S 1]{MR1954396}),
    we deduce the upper bound
    \begin{align*}
        |f'^t\cdot(\vec{u\0}-\vec{w\0})|
        &=\left|\frac{h^t\cdot (\vec{u}-\vec{w})}{-F_{2k+1}}\right|
         \leq \Vert\vec{u}-\vec{w}\Vert_\infty\frac{\Vert h\Vert_1 }{F_{2k+1}}
         = \Vert\vec{u}-\vec{w}\Vert_\infty\frac{\sum_{i=0}^{2k-1}F_i}{F_{2k+1}}\\
            &= \Vert\vec{u}-\vec{w}\Vert_\infty\frac{F_{2k+1}-2}{F_{2k+1}}
             <\Vert\vec{u}-\vec{w}\Vert_\infty \leq 2.
    \end{align*}
    Since $\valF(u\0)-\valF(w\0)=f'^t\cdot(\vec{u\0}-\vec{w\0})$ is an integer,
    it follows that
    $\valF(u\0)-\valF(w\0)\in\{-1,0,+1\}$.

    (ii) Let $(\vec{u}-\vec{w})=(d_0,\dots,d_{2k-1})$ for some integer $k\geq2$.
    By hypothesis, we have $d_{2k-3},d_{2k-2},d_{2k-1}\geq0$. 
    Using the identity $\sum_{i=0}^{2k-4}F_i = F_{2k-2}-2$,
    we compute
    \begin{align*}
        h^t\cdot (\vec{u}-\vec{w}) 
        &= \sum_{i=0}^{2k-1}(-1)^iF_id_i
        = \sum_{i=0}^{2k-4}(-1)^iF_id_i
            - F_{2k-3}d_{2k-3}
            + F_{2k-2}d_{2k-2}
            - F_{2k-1}d_{2k-1}\\
        &\leq \Vert\vec{u}-\vec{w}\Vert_\infty\left(\sum_{i=0}^{2k-4}F_i + F_{2k-2}\right)
         \leq \Vert\vec{u}-\vec{w}\Vert_\infty\left(2F_{2k-2}-2\right)
         < 4F_{2k-2}.
    \end{align*}
    Therefore, using \eqref{eq:valu0-valw0},
    \begin{align*}
        \valF(u\0)-\valF(w\0)
            = \frac{-h^t\cdot (\vec{u}-\vec{w})}{F_{2k+1}}
            > \frac{-4F_{2k-2}}{F_{2k+1}}
            = \frac{-F_{2k+1}+F_{2k-5}}{F_{2k+1}}
            > -1.
    \end{align*}

    (iii) Let $(\vec{u}-\vec{w})=(d_0,\dots,d_{2k-1})$ for some integer $k\geq2$.
    By hypothesis, we have $d_{2k-2}\geq0$ and $d_{2k-1},d_{2k-3}\in\{-1,0,1\}$.
    We compute
    \begin{align*}
        h^t\cdot (\vec{u}-\vec{w}) 
        &= \sum_{i=0}^{2k-1}(-1)^iF_id_i
        = \sum_{i=0}^{2k-4}(-1)^iF_id_i
            - F_{2k-3}d_{2k-3}
            + F_{2k-2}d_{2k-2}
            - F_{2k-1}d_{2k-1}\\
        &\geq -\Vert\vec{u}-\vec{w}\Vert_\infty
               \sum_{i=0}^{2k-4}F_i 
               - F_{2k-3}
               - F_{2k-1}\\
        &\geq -2 (F_{2k-2}-2)
               - F_{2k-3}
               - F_{2k-1}
        = -F_{2k+1}+4.
    \end{align*}
    Therefore, using \eqref{eq:valu0-valw0},
    \[
        \valF(u\0)-\valF(w\0)
            = \frac{-h^t\cdot (\vec{u}-\vec{w})}{F_{2k+1}}
            \leq \frac{F_{2k+1}-4}{F_{2k+1}}
            < 1.\qedhere
    \]
\end{proof}

This lemma implies that there is a unique way of translating
a ternary word $ua\in\Aldel^*$, with $a\in\Aldel$, into a binary word
$wbt\in\Alsig^*$ with the same Fibonacci value, provided that $w$ is a prefix of
the translation of $u$ itself.
More precisely, let $\Sigma=\{\0,\1\}$ and $\Scal$ be the set 
\[
\Scal= \Sigma^3\setminus\Sigma^*\1\1\Sigma^* 
     = \{\0\0\0, \0\0\1, \0\1\0, \1\0\0, \1\0\1\}.
\]

\begin{lemma}\label{lem:tree-structure}
    Let $u \in \Aldel^*$,
    $w \in \Alsig^*$ and
    $s\in \Scal$ be such that $|u|=|w|$ and 
    $\valF (u) = \valF (w s)$.
    For every $a\in \Aldel$, there exist a unique
    $b\in\Alsig$ and a unique $t\in\Scal$ such that
    $\valF (u a) = \valF (w b t)$.
\end{lemma}

\begin{proof}
    Since $\valF (u) = \valF (w s)$,
    from Lemma~\ref{lem:append0-change-at-most-1}, we have that
            $\valF(u\0)-\valF(ws\0)=\gamma$ for some
            $\gamma\in\{-1,0,+1\}$.
    We compute
    \begin{align*}
            \valF(u a) - \valF(w\0\0\0\0)
            &= \valF(u\0) +  a - \valF(w\0\0\0\0)\\
            &= \valF(ws\0) + \gamma +  a - \valF(w\0\0\0\0) \\
            &= \valF(s\0) + \gamma +  a.
    \end{align*}
    If $s\in\{\0\0\1,\0\1\0,\1\0\0\}$, then
    $\valF(s\0) \in \{2,3,5\}$
    and
    $\valF(s\0)+\gamma \in \{1,2,3,4,5,6\}$.
    If $s=\0\0\0$, then
    $\valF(s\0)=0$ and
    $\gamma\in\{0,+1\}$ 
    from Lemma~\ref{lem:append0-change-at-most-1} (ii)
    so that
    $\valF(s\0)+\gamma \in \{0,1\}$.
    Finally if $s=\1\0\1$, then
    $\valF(s\0)=7$ and
    $\gamma\in\{-1,0\}$
    from Lemma~\ref{lem:append0-change-at-most-1} (iii)
    so that
    $\valF(s\0)+\gamma \in \{6,7\}$.
    In summary if $s\in\Scal$, we have $\valF(s\0)+\gamma \in \{0,\dots, 7\}$
    and $\valF(s\0)+\gamma +  a \in \{0,\dots,9\}$.  
    The map $\valF$ induces a bijection
    $\0\Scal \to \{0,1,2,3,4\}$
    and a bijection
    $\1\Scal \to \{5,6,7,8,9\}$.
    Therefore there exist a unique $b \in \{\0,\1\}$ and a unique $t\in\Scal$
    such that $\valF(s\0)+\gamma +  a = \valF(b t)$,
    which holds if and only if
    $ \valF(u a) = \valF(w\0\0\0\0) + \valF(b t)
            = \valF(wb t)$.
\end{proof}

This implies the existence of a canonical translation of ternary words into
binary words having the same value.

\begin{proposition}\label{prop:uniqueness-new}
    There exist a unique 
    length-preserving map
    $w:\{\0,\1,\2\}^*\to\{\0,\1\}^*$
    and a unique map
    $s:\{\0,\1,\2\}^*\to\Scal$,
    that we denote by $w_u:=w(u)$
    and $s_u:=s(u)$,
    such that for every $u \in \{\0,\1,\2\}^*$
    \begin{itemize}
        \item $w_u$ is a prefix of $w_{ua}$ for every $a \in \{\0,\1,\2\}$,
        \item $\valF (u) = \valF (w_u s_u)$.
    \end{itemize}
    Moreover, there exists a unique map $\lambda:\Aldel^+  \to \Alsig$
    such that $w_{ua}=w_u\lambda_{ua}$
    for every $u \in \{\0,\1,\2\}^*$ and $a \in \{\0,\1,\2\}$.
\end{proposition}

\begin{proof}
    Let $u\in\Aldel^*$. 
    We carry out the proof of existence and uniqueness by induction on the length of $u$.
    If $|u|=0$, then $w_u=\varepsilon$ is the empty word, the only word of
    length 0.
    Also $s_u=\texttt{000}$ has the property that
    $\valF(u)=\valF(\varepsilon)=0=\valF(\texttt{000})=\valF (w_u s_u)$.
    Suppose that $s'_u\in\Scal$ also satisfies the equality
        $\valF (u) = \valF (w_u s'_u)$. We have
    \[
        \valF (s_u) = \valF (w_u s_u) = \valF(u) = 0 = \valF (u) = \valF (w_u s'_u)
        = \valF (s'_u)
    \]
    implies that $s_u=\texttt{000}=s'_u$.

    (Induction hypothesis).
    Let $k\in\N$.
    Suppose that for every $u\in\Aldel^*$, with $|u|\leq k$,
    there exist a unique word $w_u\in\{\0,\1\}^*$, with $|w_u|=|u|$,
    and a unique word $s_u\in\Scal$ such that
    $w_u$ is a prefix of $w_{ua}$ for every $a\in\{\0,\1,\2\}$
    and $\valF (u) = \valF (w_u s_u)$.
    Let $v\in\Aldel^{*}$ with $|v|=k+1$.
    Let $u\in\Aldel^{k}$ and $a\in\Aldel$ such that $v=ua$.

    (Existence).
    From Lemma~\ref{lem:tree-structure},
    there exist a unique
    $b\in\Alsig$ and a unique $t\in\Scal$ such that
    $\valF(ua) = \valF (w_u b t)$.
    Therefore, we let 
    $w_{ua}=w_u b$ and $s_{ua}=t$.
    It holds that
    $w_u$ is a prefix of $w_{ua}$ and
    $\valF (ua) = \valF (w_{ua} s_{ua})$.

    (Uniqueness).
    Suppose that $w_{ua},w'_{ua}\in\{\0,\1\}^{k+1}$
    and $s_{ua},s'_{ua}\in\Scal$
    satisfy the hypothesis.
    We want to show that $w_{ua}=w'_{ua}$ and $s_{ua}=s'_{ua}$.
    From the hypothesis, $w_u$ is a prefix of $w_{ua}$ and of $w'_{ua}$,
    thus let $c,c'\in\Alsig$ be such that
    $w_{ua}=w_u c$ and
    $w'_{ua}=w_u c'$.
    By the definition of $w_{ua}$, $w'_{ua}$, $s_{ua}$ and $s'_{ua}$, we have that
    \[
        \valF (w_{u}c s_{ua})
        = \valF (w_{ua} s_{ua})
        = \valF (ua) 
        = \valF (w'_{ua} s'_{ua})
        = \valF (w_{u}c' s'_{ua}).
    \]
    From Lemma~\ref{lem:tree-structure},
    there exist a unique
    $b\in\Alsig$ and a unique $t\in\Scal$ such that
    $\valF(ua) = \valF (w_u b t)$.
    Therefore, we have $c=c'$ and
    $s_{ua}=s'_{ua}$.
    We conclude that
    $w_{ua}=w_uc=w_u c'=w'_{ua}$
    and
    $s_{ua}=s'_{ua}$.

    The existence and uniqueness of the map $\lambda$ follows
    from the existence and uniqueness of the map $w$.
\end{proof}

Therefore we have that
	for every $u \in \{\0,\1,\2\}^*$ and every $a\in\Aldel$,
    the following equation holds
\begin{equation}\label{eq:uniqueness_reformulated}
		\valF(ua) 
        = \valF(w_{ua} s_{ua})
        = \valF(w_u \lambda_{ua} s_{ua}).
\end{equation}

The values of $w_u$, $\lambda_{ua}$ and $s_{ua}$ and are shown in the table below for
small words $u\in\Aldel^*$ and $a\in\Aldel$.
\[
    \small
    \begin{array}{c|llc}
        u\cdot a   & w_u &\lambda_{ua} & s_{ua}  \\
\hline
  \varepsilon\cdot\0      & \varepsilon   & \0  & \0\0\0  \\
  \varepsilon\cdot\1      & \varepsilon   & \0  & \0\0\1  \\
  \varepsilon\cdot\2      & \varepsilon   & \0  & \0\1\0  \\
\hline
  \0\cdot\0    &  \0 & \0  & \0\0\0 \\
  \0\cdot\1    &  \0 & \0  & \0\0\1 \\
  \0\cdot\2    &  \0 & \0  & \0\1\0 \\
\hline
  \1\cdot\0    &  \0 & \0  & \0\1\0 \\
  \1\cdot\1    &  \0 & \0  & \1\0\0 \\
  \1\cdot\2    &  \0 & \0  & \1\0\1 \\
\hline
  \2\cdot\0    &  \0 & \0  & \1\0\1 \\
  \2\cdot\1    &  \0 & \1  & \0\0\0 \\
  \2\cdot\2    &  \0 & \1  & \0\0\1 \\
\hline
\end{array}
\qquad
\qquad
\begin{array}{c|llc}
    u\cdot a   & w_u &\lambda_{ua} & s_{ua} \\
\hline
  \vdots & &&\\
\hline
  \0\2\cdot\0 &\texttt{00} & \0 & \texttt{101} \\
  \0\2\cdot\1 &\texttt{00} & \1 & \texttt{000} \\
  \0\2\cdot\2 &\texttt{00} & \1 & \texttt{001} \\
\hline
  \1\0\cdot\0 &\texttt{00} & \0 & \texttt{100} \\
  \1\0\cdot\1 &\texttt{00} & \0 & \texttt{101} \\
  \1\0\cdot\2 &\texttt{00} & \1 & \texttt{000} \\
\hline
  \vdots & &&\\
\hline
  \2\2\cdot\0 &\texttt{01} & \0 & \texttt{010}  \\
  \2\2\cdot\1 &\texttt{01} & \0 & \texttt{100}  \\
  \2\2\cdot\2 &\texttt{01} & \0 & \texttt{101}  \\
\hline
  \vdots & &&\\
\end{array}
\]
It can be observed that some words behave the same, for example $u=\1$ and $v=\2\2$.
This is formalized in the following definition.

\begin{definition}[equivalence relation]\label{def:equivalence}
    Let $u, v\in \Aldel^*$.
    We say that $u$ and $v$ are equivalent,
    denoted by $u \equiv v$, if
    $s_{u} = s_{v}$ and
    $\lambda_{ua}=\lambda_{va}$ and $s_{ua}=s_{va}$
    for every $a\in\Aldel$.
\end{definition}

For example, for $\1\equiv\2\2$, as we have
$s_\1=\0\0\1=s_{\2\2}$ and for every $a\in\Aldel$, $s_{\1a}=s_{\2\2a}$ and
$\lambda_{\1a}=\lambda_{\2\2a}$.
However, $\0\2\not\equiv\1\0$.
Indeed, we have
$s_{\0\2}=\0\1\0=s_{\1\0}$,
but 
$(s_{\0\2\0},
s_{\0\2\1},
s_{\0\2\2})
=
(\texttt{101},
\texttt{000},
\texttt{001})
\neq
(\texttt{100},
\texttt{101},
\texttt{000})
=
(s_{\1\0\0},
s_{\1\0\1},
s_{\1\0\2})$.

The equivalence relation can be characterized
by a map computing the defect in the
translation of the word $u\0$ into $w_u\0\0\0\0$. 
Its value allows us to prove that two ternary words behave the same.

\begin{definition}\label{def:Delta-map}
    We define $\Delta:\Aldel^*\to\Z$ 
    for every $u\in \Aldel^*$
    as
    \begin{equation*}
        \Delta(u) = \valF(u\0) - \valF(w_u \0\0\0\0).
    \end{equation*}
\end{definition}

\begin{lemma}\label{lem:expression-for-Delta}
    For every $u\in \Aldel^*$
    and every $a\in\Aldel$, we have
    \[
        \Delta(u) = \valF(\lambda_{ua}s_{ua}) - a.
    \]
\end{lemma}

\begin{proof}
	For every $u \in \{\0,\1,\2\}^*$ and every $a\in\Aldel$, 
    using \eqref{eq:uniqueness_reformulated},
    we have
    \[
    \valF(u a) 
        = \valF(w_u \lambda_{ua} s_{ua}) 
        = \valF(w_u \0\0\0\0)
        + \valF(\lambda_{ua} s_{ua}).
    \]
    Thus, we obtain
    \begin{align*}
        \Delta(u)
        &= \valF(u\0) - \valF(w_u \0\0\0\0)\\
        &= \valF(u a) - \valF(w_u \0\0\0\0) - a\\
        &= \valF(\lambda_{ua} s_{ua})-a.\qedhere
    \end{align*}
\end{proof}

\begin{lemma}\label{lem:equivalence-vs-Delta}
    Let $u,v\in \Aldel^*$.
    We have $u\equiv v$ if and only if $s_u=s_v$ and $\Delta(u) = \Delta(v)$.
\end{lemma}

\begin{proof}
    Suppose that $u\equiv v$. Then $s_u=s_v$
    and for every $a\in\Aldel$,
    we have
    $\lambda_{ua}=\lambda_{va}$ and $s_{ua}=s_{va}$.
    Thus, using Lemma~\ref{lem:expression-for-Delta}, we have
    \[
        \Delta(u)
        = \valF(\lambda_{ua} s_{ua}) - a
        = \valF(\lambda_{va} s_{va}) - a
        = \Delta(v).
    \]

    Reciprocally,
    suppose that $s_u=s_v$ and $\Delta(u) = \Delta(v)$.
    Let $a\in\Aldel$.
    From Lemma~\ref{lem:expression-for-Delta}, we have
    \[
        \valF(\lambda_{ua} s_{ua}) 
        = \Delta(u) + a
        = \Delta(v) + a
        = \valF(\lambda_{va} s_{va}).
    \]
    We have 
    $ \valF(\lambda_{ua} s_{ua})
    = \valF(\lambda_{va} s_{va})
    \in \{0,\dots,9\}$.
    From the injectivity of $\valF:\{\0,\1\}\Scal\to\{0,1,2,3,4,5,6,7,8,9\}$,
    it follows that 
    $\lambda_{ua}=\lambda_{va}$ and 
    $s_{ua}=s_{va}$.
    Thus $u\equiv v$.
\end{proof}

In the next lemma, we prove that two words that are equivalent have equivalent
children.

\begin{lemma}\label{lem:children-inherit-equivalence}
    Let $u,v\in \Aldel^*$ and $a \in\Aldel$.
    If $u \equiv v$, then $ua \equiv va$.
\end{lemma}

\begin{proof}
    Let $u,v\in \Aldel^*$ and $a \in\Aldel$.
    From Definition~\ref{def:Delta-map}, we deduce the following
    identity
    \begin{equation}\label{eq:identity-for-Delta-ua}
    \begin{aligned}
        \Delta(ua) 
        &= \valF(ua\0)-\valF(w_{ua}\0\0\0\0)\\
        &= \left(\valF(ua)+\valF(u)+a\right)
          -\left(\valF(w_{ua}\0\0\0)+\valF(w_{ua}\0\0)\right)\\
        &= \valF(w_{ua}s_{ua})+\valF(w_u s_u)+a
          -\left(\valF(w_{ua}\0\0\0)+\valF(w_{u}\lambda_{ua}\0\0)\right)\\
        &= \valF(s_{ua})+\valF(s_u)+a
          -\valF(\lambda_{ua}\0\0).
    \end{aligned}
    \end{equation}
    The equivalence $u \equiv v$ implies $s_{u}=s_{v}$,
    $\lambda_{ua}=\lambda_{va}$ and 
    $s_{ua}=s_{va}$.
    Thus, using \eqref{eq:identity-for-Delta-ua}, we have
        \begin{align*}
            \Delta(u a) 
            &= \valF(s_{ua}) + \valF(s_{u}) + a -\valF(\lambda_{ua}\0\0)\\
            &= \valF(s_{va}) + \valF(s_{v}) + a -\valF(\lambda_{va}\0\0)
            = \Delta(va).
        \end{align*}
    We conclude from Lemma~\ref{lem:equivalence-vs-Delta} that
    $ua \equiv va$.
\end{proof}

As there are finitely many combinations of states $s_u\in\Scal$ and values
$\Delta(u)=\valF(\lambda_{u\0}s_{u\0})\in\{0,\dots,9\}$, 
the quotient $\Aldel^*|_\equiv$
contains at most $\#\Scal\cdot 10=50$ equivalence classes.
In each equivalence class, there is a word which is minimal with respect to
the radix order. The radix order is defined as $u<_{\rad}v$ either if $u$ is shorter than $v$,
or if $u$ and $v$ have the same length and $u$ is lexicographically smaller than $v$.
In the following proof, we build a graph exploring all equivalence classes
that are accessible from the equivalence class of the empty word,
for which $s_\varepsilon=\0\0\0$ and $\Delta(\varepsilon)=0$.

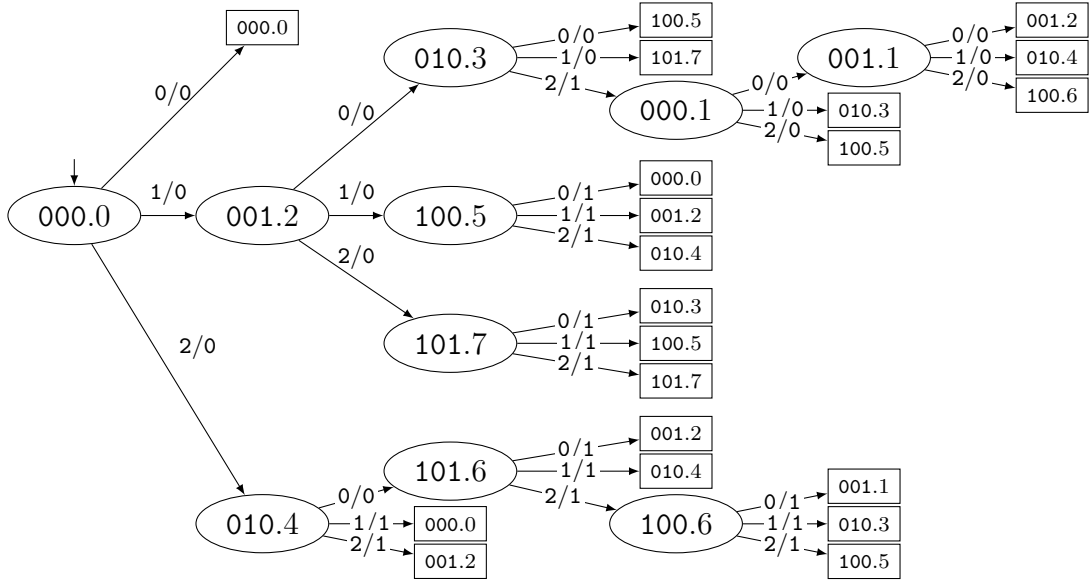
\begin{figure}[h]
\begin{center}
    \begin{tikzpicture}[>=latex,yscale=1.5,xscale=0.95]
        \tikzstyle{rond}=[draw,ellipse]
        \tikzstyle{carre}=[draw,rectangle,font=\tiny]
        
    \begin{scope}
        \node [coordinate] (A) at (0,.5){};
        \node (0) [rond] at (0,0) {\texttt{000}.0};
        \node[right of=0,node distance=25mm,yshift=25mm ] (0L) [carre]      {\texttt{000}.0};
        \node[right of=0,node distance=25mm            ] (1) [rond]         {\texttt{001}.2};
        \node[right of=0,node distance=25mm,yshift=-41mm] (2) [rond]        {\texttt{010}.4};
        \node[right of=1,node distance=25mm,yshift=21mm ] (10) [rond]       {\texttt{010}.3};
        \node[right of=1,node distance=25mm            ] (11) [rond]        {\texttt{100}.5};
        \node[right of=1,node distance=25mm,yshift=-17mm] (12) [rond]       {\texttt{101}.7};
        \node[right of=2,node distance=25mm,yshift=7mm ] (20) [rond]        {\texttt{101}.6};
        \node[right of=2,node distance=25mm            ] (21) [carre]       {\texttt{000}.0};
        \node[right of=2,node distance=25mm,yshift=-5mm] (22) [carre]       {\texttt{001}.2};
        \node[right of=10,node distance=30mm,yshift=5mm ] (100) [carre]     {\texttt{100}.5};
        \node[right of=10,node distance=30mm            ] (101) [carre]     {\texttt{101}.7};
        \node[right of=10,node distance=30mm,yshift=-7mm] (102) [rond]      {\texttt{000}.1};
        \node[right of=11,node distance=30mm,yshift=5mm ] (110) [carre]     {\texttt{000}.0};
        \node[right of=11,node distance=30mm            ] (111) [carre]     {\texttt{001}.2};
        \node[right of=11,node distance=30mm,yshift=-5mm] (112) [carre]     {\texttt{010}.4};
        \node[right of=12,node distance=30mm,yshift=5mm ] (120) [carre]     {\texttt{010}.3};
        \node[right of=12,node distance=30mm            ] (121) [carre]     {\texttt{100}.5};
        \node[right of=12,node distance=30mm,yshift=-5mm] (122) [carre]     {\texttt{101}.7};
        \node[right of=20,node distance=30mm,yshift=5mm ] (200) [carre]     {\texttt{001}.2};
        \node[right of=20,node distance=30mm            ] (201) [carre]     {\texttt{010}.4};
        \node[right of=20,node distance=30mm,yshift=-7mm] (202) [rond]      {\texttt{100}.6};
        \node[right of=102,node distance=25mm,yshift=7mm ] (1020) [rond]    {\texttt{001}.1};
        \node[right of=102,node distance=25mm            ] (1021) [carre]   {\texttt{010}.3};
        \node[right of=102,node distance=25mm,yshift=-5mm] (1022) [carre]   {\texttt{100}.5};
        \node[right of=202,node distance=25mm,yshift=5mm ] (2020) [carre]   {\texttt{001}.1};
        \node[right of=202,node distance=25mm            ] (2021) [carre]   {\texttt{010}.3};
        \node[right of=202,node distance=25mm,yshift=-5mm] (2022) [carre]   {\texttt{100}.5};
        \node[right of=1020,node distance=25mm,yshift=5mm ] (10200) [carre] {\texttt{001}.2};
        \node[right of=1020,node distance=25mm            ] (10201) [carre] {\texttt{010}.4};
        \node[right of=1020,node distance=25mm,yshift=-5mm] (10202) [carre] {\texttt{100}.6};
    \draw[->] (A)  --  (0);

    \draw[->] (0)  -- node[above]{\scriptsize $\0/\0$} (0L);
    \draw[->] (0)  -- node[above]{\scriptsize $\1/\0$} (1);
    \draw[->] (0)  -- node[above right]{\scriptsize $\2/\0$} (2);
    \draw[->] (1)  -- node[above]{\scriptsize $\0/\0$} (10);
    \draw[->] (1)  -- node[above]{\scriptsize $\1/\0$} (11);
    \draw[->] (1)  -- node[above]{\scriptsize $\2/\0$} (12);
    \draw[->] (2)  -- node[fill=white,inner sep=1pt]{\scriptsize $\0/\0$} (20);
    \draw[->] (2)  -- node[fill=white,inner sep=1pt]{\scriptsize $\1/\1$} (21);
    \draw[->] (2)  -- node[fill=white,inner sep=1pt]{\scriptsize $\2/\1$} (22);
    \draw[->] (10)  -- node[fill=white,inner sep=1pt]{\scriptsize $\0/\0$} (100);
    \draw[->] (10)  -- node[fill=white,inner sep=1pt]{\scriptsize $\1/\0$} (101);
    \draw[->] (10)  -- node[fill=white,inner sep=1pt]{\scriptsize $\2/\1$} (102);
    \draw[->] (11)  -- node[fill=white,inner sep=1pt]{\scriptsize $\0/\1$} (110);
    \draw[->] (11)  -- node[fill=white,inner sep=1pt]{\scriptsize $\1/\1$} (111);
    \draw[->] (11)  -- node[fill=white,inner sep=1pt]{\scriptsize $\2/\1$} (112);
    \draw[->] (12)  -- node[fill=white,inner sep=1pt]{\scriptsize $\0/\1$} (120);
    \draw[->] (12)  -- node[fill=white,inner sep=1pt]{\scriptsize $\1/\1$} (121);
    \draw[->] (12)  -- node[fill=white,inner sep=1pt]{\scriptsize $\2/\1$} (122);
    \draw[->] (20)  -- node[fill=white,inner sep=1pt]{\scriptsize $\0/\1$} (200);
    \draw[->] (20)  -- node[fill=white,inner sep=1pt]{\scriptsize $\1/\1$} (201);
    \draw[->] (20)  -- node[fill=white,inner sep=1pt]{\scriptsize $\2/\1$} (202);
    \draw[->] (102)  -- node[fill=white,inner sep=1pt]{\scriptsize $\0/\0$} (1020);
    \draw[->] (102)  -- node[fill=white,inner sep=1pt]{\scriptsize $\1/\0$} (1021);
    \draw[->] (102)  -- node[fill=white,inner sep=1pt]{\scriptsize $\2/\0$} (1022);
    \draw[->] (202)  -- node[fill=white,inner sep=1pt]{\scriptsize $\0/\1$} (2020);
    \draw[->] (202)  -- node[fill=white,inner sep=1pt]{\scriptsize $\1/\1$} (2021);
    \draw[->] (202)  -- node[fill=white,inner sep=1pt]{\scriptsize $\2/\1$} (2022);
    \draw[->] (1020)  -- node[fill=white,inner sep=1pt]{\scriptsize $\0/\0$} (10200);
    \draw[->] (1020)  -- node[fill=white,inner sep=1pt]{\scriptsize $\1/\0$} (10201);
    \draw[->] (1020)  -- node[fill=white,inner sep=1pt]{\scriptsize $\2/\0$} (10202);
    \end{scope}
    \end{tikzpicture}
    \caption{The edges represent the directed labeled graph $G$, which
        can be folded (after merging equivalent states) into the
        Mealy machine $\Zcal$ equivalent to the Berstel adder $\Bcal$.
        A vertex reached from a path $u$ has an ellipse shape if and only if
        $u$ is minimal with respect to the radix order within the equivalence
        class $[u]_\equiv$ and has a rectangle shape if $u\equiv v$ for some
        word $v<_{\rad} u$.}
    \label{fig:tree}
\end{center}
\end{figure}

\begin{proof}[Proof of Theorem~\ref{thm:Berstel}]
Let $G$ be the directed labeled graph 
$G = (V, E)$ where 
    \begin{align*}
        V &= \{[u]_\equiv\colon u\in\Aldel^*\},\\
        E &=\{([u]_\equiv,[ua]_\equiv)\in V\times V\colon u\in\Aldel^*,a\in\Aldel\},
    \end{align*}
which is well-defined
following Lemma~\ref{lem:children-inherit-equivalence}.
The connected component of the equivalence class of the empty word 
in the graph $G$ is shown in Figure \ref{fig:tree}
where vertices $[u]_\equiv$ are denoted by $s_{u}.\Delta(u)$,
following Lemma~\ref{lem:equivalence-vs-Delta},
and edges are shown as
$ [u]_\equiv \xrightarrow{a/\lambda_{ua}} [ua]_\equiv$.

It follows from the exploration of the connected component of the equivalence
of the empty word in radix order
that for every $u\in\Aldel^*$, we have $0\leq\Delta(u)\leq7$; see Figure~\ref{fig:tree}.

The cardinality of $V$ is 10 (the 10 vertices shown with an ellipse
shape in Figure~\ref{fig:tree}) and the directed labeled graph $G$ can be
equivalently described by a Mealy machine.
Let $\Zcal = (Q,S_0, A, B, \delta_\Zcal, \eta_\Zcal, \phi_\Zcal)$ be the Mealy machine
such that 
\begin{itemize}
    \item $Q = \{s_u.\Delta(u)\colon [u]_\equiv\in V\}$,
    $S_0 = \0\0\0.0$, $A = \Aldel$, $B = \Alsig$,
    \item for every $s_{u}.\Delta(u)\in Q$ and $a \in A$, 
        $\delta_\Zcal(s_u.\Delta(u),a) = s_{ua}.\Delta(u a)$  
        and $\eta_\Zcal(s_u.\Delta(u),a) = \lambda_{ua}$.
    \item for every $s_{u}.\Delta(u)\in Q$,
        $\phi_\Zcal(s_u.\Delta(u)) = s_{u}$.
\end{itemize}
The Mealy machine $\Zcal$ computes $w_u$ and $s_u$ in the sense that $\Zcal(u)=w_u$ and
    $\Zcal_\downarrow(u)=s_u$ for every word $u\in\Aldel^*$.
    Therefore, if $u = \varepsilon$, then $\Zcal(u) = \varepsilon$,
    $\Zcal_{\downarrow}(u) = \0\0\0$ and 
    \[
      \valF(u)
        = \valF(\varepsilon) 
        = 0 
        = \valF(\0\0\0) 
        = \valF(\Zcal(u)\cdot \Zcal_{\downarrow}(u)).
    \]
    If $u\in\Aldel^*$ and $a\in\Aldel$,
    then $\Zcal(u a) = w_u \lambda_{ua}$ and $\Zcal_{\downarrow}(u a)=s_{ua}$,
    and by Equation~\eqref{eq:uniqueness_reformulated}, 
    \[
    \valF(u a)
    = \valF(w_u \lambda_{ua} s_{ua})
    = \valF(\Zcal(u a) \cdot \Zcal_{\downarrow}(u a)).
    \]
    This ends the proof as the machine $\Zcal$ is equivalent to the Berstel
    adder $\Bcal$ reproduced in Figure~\ref{fig:berstel_transducer}.
\end{proof}

\bibliographystyle{myalpha} %

\bibliography{biblio}

\begin{thebibliography}{AUFP13}

\bibitem[AUFP13]{MR3093678}
C. Ahlbach, J. Usatine, C. Frougny, and N. Pippenger.
\newblock Efficient algorithms for {Z}eckendorf arithmetic.
\newblock {\em Fibonacci Quart.}, 51(3):249--255, 2013.

\bibitem[Ber86]{zbMATH03947643}
J. Berstel.
\newblock Fibonacci words - a survey.
\newblock In {\em The book of L}, pages 13--27. 1986.
\newblock dedic. A. Lindenmayer Occas. 60th Birthday.

\bibitem[BR10]{MR2742574}
V. Berth\'e and M. Rigo, editors.
\newblock {\em Combinatorics, automata and number theory}, volume 135 of {\em
  Encyclopedia of Mathematics and its Applications}.
\newblock Cambridge University Press, Cambridge, 2010.
\newblock
  \href{https://doi.org/10.1017/CBO9780511777653}{\texttt{doi:10.1017/CBO9780511777653}}.

\bibitem[Bun92]{MR1162409}
M.~W. Bunder.
\newblock Zeckendorf representations using negative {F}ibonacci numbers.
\newblock {\em Fibonacci Quart.}, 30(2):111--115, 1992.

\bibitem[Car68]{MR236094}
L. Carlitz.
\newblock Fibonacci representations.
\newblock {\em Fibonacci Quart.}, 6(4):193--220, 1968.

\bibitem[Day60]{MR112863}
D.~E. Daykin.
\newblock Representation of natural numbers as sums of generalised {F}ibonacci
  numbers.
\newblock {\em J. London Math. Soc.}, 35:143--160, 1960.
\newblock
  \href{https://doi.org/10.1112/jlms/s1-35.2.143}{\texttt{doi:10.1112/jlms/s1-35.2.143}}.

\bibitem[Dic66]{MR0245499}
L.~E. Dickson.
\newblock {\em History of the theory of numbers. {V}ol. {I}: {D}ivisibility and
  primality}.
\newblock Chelsea Publishing Co., New York, 1966.

\bibitem[Fra85]{MR777556}
A.~S. Fraenkel.
\newblock Systems of numeration.
\newblock {\em Amer. Math. Monthly}, 92(2):105--114, 1985.
\newblock \href{https://doi.org/10.2307/2322638}{\texttt{doi:10.2307/2322638}}.

\bibitem[Fro88]{MR942576}
C. Frougny.
\newblock Linear numeration systems of order two.
\newblock {\em Inform. and Comput.}, 77(3):233--259, 1988.
\newblock
  \href{https://doi.org/10.1016/0890-5401(88)90050-8}{\texttt{doi:10.1016/0890-5401(88)90050-8}}.

\bibitem[Fro91]{MR1093759}
C. Frougny.
\newblock Fibonacci representations and finite automata.
\newblock {\em IEEE Trans. Inform. Theory}, 37(2):393--399, 1991.
\newblock
  \href{https://doi.org/10.1109/18.75263}{\texttt{doi:10.1109/18.75263}}.

\bibitem[Fro99]{MR1705857}
C. Frougny.
\newblock On-line finite automata for addition in some numeration systems.
\newblock {\em Theor. Inform. Appl.}, 33(1):79--101, 1999.
\newblock
  \href{https://doi.org/10.1051/ita:1999107}{\texttt{doi:10.1051/ita:1999107}}.

\bibitem[FS96]{MR1411227}
C. Frougny and B. Solomyak.
\newblock On representation of integers in linear numeration systems.
\newblock In {\em Ergodic theory of {${\bf Z}^d$} actions ({W}arwick,
  1993--1994)}, volume 228 of {\em London Math. Soc. Lecture Note Ser.}, pages
  345--368. Cambridge Univ. Press, Cambridge, 1996.
\newblock
  \href{https://doi.org/10.1017/CBO9780511662812.014}{\texttt{doi:10.1017/CBO9780511662812.014}}.

\bibitem[FS10]{MR2766740}
C. Frougny and J. Sakarovitch.
\newblock Number representation and finite automata.
\newblock In {\em Combinatorics, automata and number theory}, volume 135 of
  {\em Encyclopedia Math. Appl.}, pages 34--107. Cambridge Univ. Press,
  Cambridge, 2010.

\bibitem[GKP94]{MR1397498}
R.~L. Graham, D.~E. Knuth, and O. Patashnik.
\newblock {\em Concrete mathematics}.
\newblock Addison-Wesley Publishing Company, Reading, MA, second edition, 1994.
\newblock A foundation for computer science.

\bibitem[Hog69]{zbMATH03316160}
V.~E.~j. Hoggatt.
\newblock Fibonacci and {Lucas} numbers.
\newblock Boston etc.: {Houghton} {Mifflin} {Company} {IV}, 92 p. (1969).,
  1969.

\bibitem[Knu97]{MR3077152}
D.~E. Knuth.
\newblock {\em The art of computer programming. {V}ol. 1}.
\newblock Addison-Wesley, Reading, MA, 1997.
\newblock Fundamental algorithms, Third edition.

\bibitem[Knu98]{MR3077153}
D.~E. Knuth.
\newblock {\em The art of computer programming. {V}ol. 2}.
\newblock Addison-Wesley, Reading, MA, 1998.
\newblock Seminumerical algorithms, Third edition.

\bibitem[Knu11]{MR3444818}
D.~E. Knuth.
\newblock {\em The art of computer programming. {V}ol. 4{A}. {C}ombinatorial
  algorithms. {P}art 1}.
\newblock Addison-Wesley, Upper Saddle River, NJ, 2011.

\bibitem[Lek52]{MR58626}
C.~G. Lekkerkerker.
\newblock Voorstelling van natuurlijke getallen door een som van getallen van
  {F}ibonacci.
\newblock {\em Simon Stevin}, 29:190--195, 1952.

\bibitem[Lin12]{zbMATH05970650}
P. Linz.
\newblock {\em An introduction to formal languages and automata}.
\newblock Sudbury, MA: Jones \& Bartlett Learning, 5th ed. edition, 2012.

\bibitem[LL21]{MR4364231}
S. Labb\'{e} and J. Lep\v{s}ov\'{a}.
\newblock A numeration system for {F}ibonacci-like {W}ang shifts.
\newblock In {\em Combinatorics on words}, volume 12847 of {\em Lecture Notes
  in Comput. Sci.}, pages 104--116. Springer, Cham, 2021.
\newblock
  \href{https://doi.org/10.1007/978-3-030-85088-3_9}{\texttt{doi:10.1007/978-3-030-85088-3\_9}}.

\bibitem[LL23]{2302.14481}
S. Labb\'{e} and J. Lep\v{s}ov\'{a}.
\newblock {D}umont-{T}homas numeration systems for $\mathbb{Z}$.
\newblock 2023.
\newblock \href{https://arxiv.org/abs/2302.14481}{arXiv:2302.14481}.

\bibitem[Lot02]{MR1905123}
M. Lothaire.
\newblock {\em Algebraic Combinatorics on Words}, volume~90 of {\em
  Encyclopedia of Mathematics and its Applications}.
\newblock Cambridge University Press, Cambridge, 2002.

\bibitem[MPR19]{MR3933318}
A. Massuir, J. Peltom\"{a}ki, and M. Rigo.
\newblock Automatic sequences based on {P}arry or {B}ertrand numeration
  systems.
\newblock {\em Adv. in Appl. Math.}, 108:11--30, 2019.
\newblock
  \href{https://doi.org/10.1016/j.aam.2019.03.003}{\texttt{doi:10.1016/j.aam.2019.03.003}}.

\bibitem[MSS16]{MR3518158}
H. Mousavi, L. Schaeffer, and J. Shallit.
\newblock Decision algorithms for {F}ibonacci-automatic words, {I}: {B}asic
  results.
\newblock {\em RAIRO Theor. Inform. Appl.}, 50(1):39--66, 2016.
\newblock
  \href{https://doi.org/10.1051/ita/2016010}{\texttt{doi:10.1051/ita/2016010}}.

\bibitem[{OEI}23]{OEISA003482}
{OEIS Foundation Inc.}
\newblock Entry {A}003482 in the on-line encyclopedia of integer sequences,
  2023.
\newblock \url{https://oeis.org/A003482}.

\bibitem[Ost22]{MR3069389}
A. Ostrowski.
\newblock Bemerkungen zur {T}heorie der {D}iophantischen {A}pproximationen.
\newblock {\em Abh. Math. Sem. Univ. Hamburg}, 1(1):77--98, 1922.
\newblock
  \href{https://doi.org/10.1007/BF02940581}{\texttt{doi:10.1007/BF02940581}}.

\bibitem[Par64]{MR166332}
W. Parry.
\newblock Representations for real numbers.
\newblock {\em Acta Math. Acad. Sci. Hungar.}, 15:95--105, 1964.
\newblock
  \href{https://doi.org/10.1007/BF01897025}{\texttt{doi:10.1007/BF01897025}}.

\bibitem[R{\'{e}}n57]{MR97374}
A. R{\'{e}}nyi.
\newblock Representations for real numbers and their ergodic properties.
\newblock {\em Acta Math. Acad. Sci. Hungar.}, 8:477--493, 1957.
\newblock
  \href{https://doi.org/10.1007/BF02020331}{\texttt{doi:10.1007/BF02020331}}.

\bibitem[Sak87]{MR906959}
J. Sakarovitch.
\newblock Easy multiplications. {I}. {T}he realm of {K}leene's theorem.
\newblock {\em Inform. and Comput.}, 74(3):173--197, 1987.
\newblock
  \href{https://doi.org/10.1016/0890-5401(87)90020-4}{\texttt{doi:10.1016/0890-5401(87)90020-4}}.

\bibitem[Sak09]{MR2567276}
J. Sakarovitch.
\newblock {\em Elements of automata theory}.
\newblock Cambridge University Press, Cambridge, 2009.
\newblock Translated from the 2003 French original by Reuben Thomas.
\newblock
  \href{https://doi.org/10.1017/CBO9781139195218}{\texttt{doi:10.1017/CBO9781139195218}}.

\bibitem[Vor02]{MR1954396}
N.~N. Vorobiev.
\newblock {\em Fibonacci numbers}.
\newblock Birkh\"{a}user Verlag, Basel, 2002.
\newblock Translated from the 6th (1992) Russian edition by Mircea Martin.
\newblock
  \href{https://doi.org/10.1007/978-3-0348-8107-4}{\texttt{doi:10.1007/978-3-0348-8107-4}}.

\bibitem[Zec72]{MR308032}
E. Zeckendorf.
\newblock Repr\'{e}sentation des nombres naturels par une somme de nombres de
  {F}ibonacci ou de nombres de {L}ucas.
\newblock {\em Bull. Soc. Roy. Sci. Li\`ege}, 41:179--182, 1972.

\end{thebibliography}

\end{document}